\documentclass[11pt,a4paper]{article}

\usepackage[utf8]{inputenc}
\usepackage[english]{babel}

\usepackage{geometry}
\usepackage{graphicx}
\usepackage{amssymb,amsmath,amsthm}
\usepackage{delarray}
\usepackage{enumitem}
\usepackage{tikz}
\usepackage{hyperref}
\graphicspath{{figures/}}
\usepackage{authblk}

\newtheorem{theorem}{\textbf{Theorem}}

\newtheorem{remark}{\textbf{Remark}}
\newtheorem*{remark*}{\textbf{Remark}}
\newtheorem{claim}{\textbf{Claim}}

\newenvironment{subproof}[1][\proofname]{%
  \begin{proof}[#1]%
}{%
  \end{proof}%
}

\newcommand{\N}{\mathbb{N}}

\newcommand{\Poly}{{\mathsf{P}}}
\newcommand{\NP}{{\mathsf{NP}}}

\newcommand{\dom}{{\mathtt{dom}}}

\newcommand{\decisionpb}[3]{\fbox{\parbox{0.9\textwidth}{{\bf #1}\\{\it Input:} #2\\{\it Question:} #3}}}

\newcommand{\sqrtthree}{1.732}
\newcommand{\cells}{\mathcal{C}}
\newcommand{\adj}{{\tikz[scale=.1]{\draw[thick] (0,-1) -- (\sqrtthree/2,-1/2) -- (\sqrtthree/2,1/2) -- (0,1) -- (-\sqrtthree/2,1/2) -- (-\sqrtthree/2,-1/2) -- cycle;}}}
\newcommand{\adjs}[1]{{\leftrightsquigarrow_{#1}}}

\newcommand{\TODO}[1]{{\color{red}\fbox{TODO:} {\sf #1}}}
\newcommand{\COMMENT}[1]{}

\title{Rikudo is NP-complete}
\author[a]{Viet-Ha Nguyen}
\author[a,b]{K{\'e}vin Perrot}
\affil[a]{Univ. C\^{o}te d'Azur, CNRS, Inria, I3S, Sophia Antipolis, France}
\affil[b]{Aix-Marseille Univ., Univ. de Toulon, CNRS, LIS, Marseille, France}
\date{}

\begin{document}
\renewcommand{\labelitemi}{$\bullet$}
\renewcommand{\labelitemii}{$\bullet$}
\setlist[itemize,enumerate]{nosep}
\maketitle

\begin{abstract}
  Rikudo is a number-placement puzzle,
  where the player is asked to complete a Hamiltonian path on a hexagonal grid,
  given some clues (numbers already placed and edges of the path).
  We prove that the game is complete for $\NP$, even if the puzzle has no hole.
  When all odd numbers are placed it is in $\Poly$,
  whereas it is still $\NP$-hard when all numbers of the form $3k+1$ are placed.
\end{abstract}

\section{Introduction}

We discovered {\em Rikudo}
in the game column of local newspaper {\em Le Progrès}.
It has been imagined a few years ago by two French guys, Paul and Xavier,
willing to offer a new challenge to {\em Sudoku} lovers\cite{rikudo}.
The study of the computational complexity of games is quite a tradition now
\cite{hanabi,hex1,checkers,candy1,mine1,gameset,go,npv20,gobang,hex2,mine2},
and we could not resist to state these results.
Even if they are not highly technical, some are not trivial.

The game consists in placing numbers from $1$ to $n$
on the cells of a hexagonal grid, so that the sequence
forms a Hamiltonian path (successive numbers must be placed on adjacent cells).
Furthermore, some numbers are already placed,
and some adjacencies are given so that the path must follow them.

In Section~\ref{s:model} we present the theoretical modelization of {\em Rikudo}
and the problems of solving such puzzles,
which are proven to be $\NP$-hard in Section~\ref{s:hardness}.
Given that the reductions do not make use of numbers already placed on the grid,
Section~\ref{s:conditions} discusses the complexity under the additional constraint
that some fixed fraction $\alpha$ of the numbers are already placed,
or that all numbers of the form $xk+1$ are already placed for some fixed integer $k$.

\section{Model}
\label{s:model}

As for {\em Sudoku}, every {\em Rikudo} instance
one finds in newspapers are solvable.
In order to model it as a decision problem in complexity theory,
we have to create some negative instances
({\em i.e.} games that do not admit any solution),
and also to create games of arbitrary size.
The original game can be played at
\url{http://www.rikudo.fr/}.

Let us consider the hexagonal grid with pointy orientation\footnote{Let us recall that
hexagonal grids follow one of two orientations: {\em pointy} or {\em flat}.},
and denote $\cells$ the set of cells.
Given a subset $\tau \subset \cells$,
we associate the loopless graph $G_\tau$ on vertex set $\tau$ with adjacency relation
$\adj$ corresponding to pairs of cells sharing an edge.
We say that $\tau$ is {\em connected} whenever $G_\tau$ is connected.
A {\em Rikudo game} is a connected finite subset of cells $\tau$ and,
with $n=|\tau|$ and $[n]=\{1,\dots,n\}$,
\begin{itemize}
  \item a partial injective map $m : \tau \to [n]$ (numbers already placed),
  \item a subset of the adjacency relation $p \subseteq \adj$ (given adjacencies).
\end{itemize}
A {\em solution} is a bijective map $s : \tau \to [n]$ such that,
with symmetric binary relation $\adjs{s}$ on $\cells$ defined as $c \,\adjs{s}\, c' \iff |s(c)-s(c')|=1$,
\begin{itemize}
  \item $\adjs{s} \subseteq \adj$
    (the sequence of numbers forms a Hamiltonian path in $G_\tau$),
  \item $s(c) = m(c)$ for all cells $c$ in the domain of $m$
    (respect given numbers),
  \item $p \subseteq \adjs{s}$
    (respect given adjacencies).
\end{itemize}
See some examples on Figure~\ref{fig:rikudo}.

\begin{figure}
  \centerline{
    \includegraphics{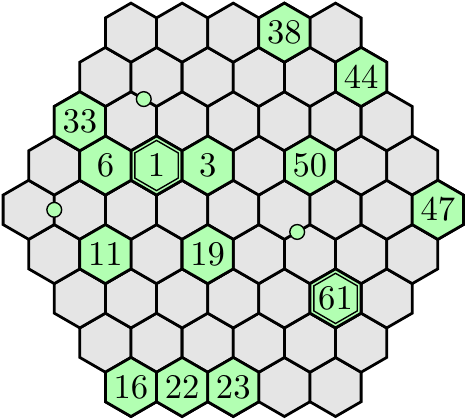}
    \hspace*{.5cm}
    \includegraphics{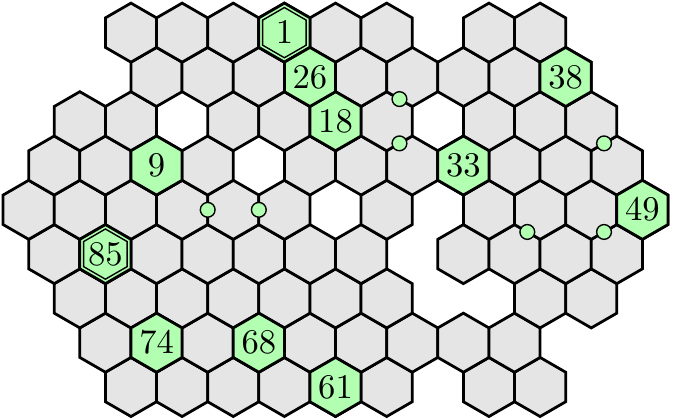}
  }
  \caption[test]{
    Two examples of {\em Rikudo} games.
    Player is asked to complete a Hamiltonian path of numbers ($\adjs{s}$)
    on a subset of cells ($\tau$) from the hexagonal grid
    (from $1$ to $n=61$ on the left, from $1$ to $n=85$ on the righ),
    given some numbers already placed ($m$),
    following edge-to-edge adjacencies ($\adj$),
    and respecting the ones indicated by a circle ($p$).
    Solutions in Appendix~\ref{a:solutions}.
  }
  \label{fig:rikudo}
\end{figure}

The main problems we are interested in are:\\[.5em]
\decisionpb{Rikudo}
{a game $(\tau,m,p)$.}
{does it admit a solution?}\\[.5em]
\decisionpb{Rikudo without holes}
{a game $(\tau,m,p)$ such that $G_{\cells \setminus \tau}$ is connected.}
{does it admit a solution?}\\[.5em]
Observe that both problems are trivially in $\NP$. 

\section{$\NP$-hardness}
\label{s:hardness}

Both reductions will be made from a closely related problem, namely of
deciding the existence of a Hamiltonian cycle in a hexagonal grid graph.
In order to avoid ambiguity, we call the six ``vertices'' of a hexagonal cell, its six {\em corners}.
A {\em hexagonal grid graph} is given by
a vertex set that is a subset of corners from a unit side length regular hexagonal tiling of the plane,
and whose edge set connects vertices that are one unit apart.
The problem is called {\bf Hamiltonian circuits in hexagonal grid graphs} ({\bf HCH}),
known to be $\NP$-complete~\cite{imnrx07}.
Without loss of generality we will consider that the instances of {\bf HCH}
have only vertices of degree two and three (no vertex of degree one).

\begin{theorem}
  \label{theorem:holes}
  {\bf Rikudo} is $\NP$-hard,
  even when $p=\emptyset$
  and $m$ has domain $\dom(m)=\{1,n\}$.
\end{theorem}

\begin{proof}
  The reduction from {\bf HCH} is almost trivial,
  by noting that any hexagonal grid graph $H$ equals some $G_\tau$.
  Note that the vertices of $H$ correspond to corners of the hexagonal cells,
  whereas the vertices of $G_\tau$ correspond to the cells themselves.
  To align two such graphs, one may consider the graph $H$
  from the hexagonal grid with flat orientation,
  scale up $H$ by a factor of $\sqrt{3}$,
  and align its vertices with the center of cells in
  the pointy hexagonal grid to play {\em Rikudo}.

  \begin{figure}
    \centerline{
      \includegraphics{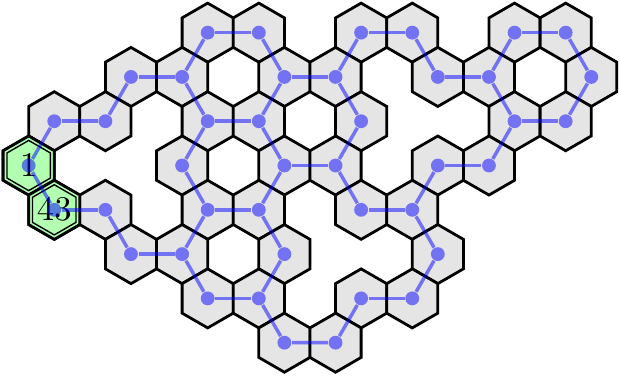}
    }
    \caption{
      Reduction from {\bf HCH} (graph in blue) to {\bf Rikudo}, with $n=43$.
    }
    \label{fig:holes}
  \end{figure}

  The reduction (see Figure~\ref{fig:holes})
  then simply consists in the game $(\tau,m,p)$
  with $\tau$ the set of cells hosting a vertex of $H$,
  $p=\emptyset$, and $m$ defined on only two cells $t,t'\in\tau$, chosen
  such that $t \,\adj\, t'$ and $t$ has degree $2$ in $G_\tau$
  (such a cell always exists), as
  $m(t)=1$ and $m(t')=n$ where $n$ is the number of vertices in $H$.
  We have $H=G_\tau$ and the Hamiltonian path of {\em Rikudo} is enforced to be a cycle by $m$,
  hence the {\bf HCH} and {\bf Rikudo} instances are identical.
\end{proof}

\COMMENT{
\begin{proof}
\textcolor{red}{alternative proof}\\
Let a hexagonal grid graph $H$ scaled up by a factor of $\sqrt{3}$ for conveniences.
The reduction from HCH to Rikudo is the following (see Figure \ref{fig:holes} for an example):
\begin{itemize}
\item a set $\tau$ is obtained by replacing each vertex $v$ of $H$ to a hexagon cell whose side is 1 where $v$ is the center of the cell. By the construction, any two adjacent vertices $u, v$ is performed by two adjacent hexagon cells (where both share a common side),
\item let $p = \emptyset$,
\item $m$ defined on only two cells $t,t'\in\tau$, chosen
  such that $t \adj t'$ and $t$ has degree $2$ in $G_\tau$
  (such a cell always exists), as
  $m(t)=1$ and $(t')=n$ where $n$ is the number of vertices of $H$.
\end{itemize}
\end{proof}
}

The restriction on $p$ and $m$ in Theorem~\ref{theorem:holes}
motivates the consideration of {\bf Rikudo without holes},
which seems closer to the spirit of the ``real'' {\em Rikudo} game
and requires more advanced constructions.

\begin{theorem}
  \label{theorem:withoutholes}
  {\bf Rikudo without holes} is $\NP$-hard.
\end{theorem}

\begin{proof}
  Let $H$ be an instance of {\bf HCH}.
  Without loss of generality we can consider that the vertices of $H$ have degree two or three
  (hexagonal grid graphs have degree at most three,
  and if $H$ has a vertex of degree one then
  it is trivially impossible to have a Hamiltonian cycle).
  We construct the game $(\tau,m,p)$ as follows (see Figure~\ref{fig:withoutholes}).
  First, take $H$ with flat orientation,
  scale it up by a factor of $2\sqrt{7}$,
  rotate it counterclockwise by an angle of $\arctan\frac{1}{3\sqrt{3}}$,
  and align its vertices with the common corners of three cells
  in the pointy hexagonal grid to play {\em Rikudo}.
  We have that each vertex $v$ of $H$ is at the corner of three cells $c_1(v),c_2(v),c_3(v)$ of the game,
  and that the middle of each edge $e$ of $H$ is at the center of a cell $c(e)$ of the game.
  Let $c_4(v)$ (respectively $c_5(v)$; $c_6(v)$) denote the cell adjacent
  to both $c_1(v),c_2(v)$ (respectively $c_2(v),c_3(v)$; $c_3(v),c_1(v)$)
  which is not $c_3(v)$ (respectively $c_1(v)$; $c_2(v)$).
  For all $v$ the six cells $c_1(v),\dots,c_6(v)$
  form an upward or downward triangle.
  Consider the set of cells $\tau'$ obtained by:
  \begin{itemize}
    \item for each vertex $v$ of $H$,
      add $c_1(v),c_2(v),c_3(v),c_4(v),c_5(v),c_6(v)$ to $\tau'$,
      and call this subset of six cells a {\em vertex gadget};
    \item for each edge $e$ of $H$,
      add the cell $c(e)$ to $\tau'$,
      and call this cell an {\em edge gadget}.
  \end{itemize}
  The cells from each finite connected component of $G_{\cells \setminus \tau'}$
  is called a {\em hole} of $\tau'$.
  The set $\tau$ is obtained from $\tau'$ by adding
  all the holes of $\tau'$.
  This ensures that $G_{\cells \setminus \tau}$ is connected.

  \begin{figure}
    \centerline{
      \includegraphics{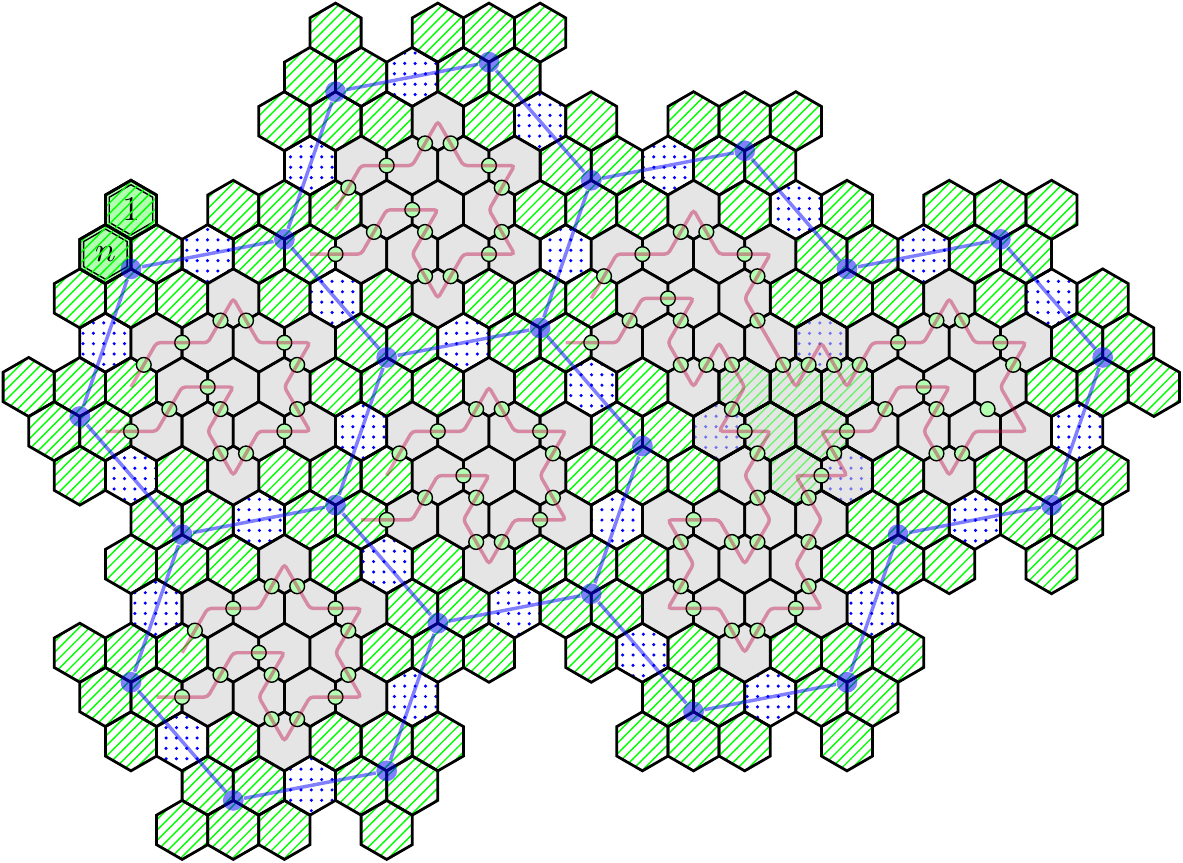}
    }
    \caption{
      Reduction from {\bf HCH} (graph in blue) to {\bf Rikudo without holes}.
      Vertex gadgets are hatched in green,
      edge gadgets are dotted in blue, holes are filled in grey
      (with the absent vertex and edge gadgets slightly marked)
      and the Hamiltonian paths of holes are highlighted in purple.
      In this example, $n=272$.
    }
    \label{fig:withoutholes}
  \end{figure}

  The set of adjacencies $p$ is build from a Hamiltonian path
  on each hole of $\tau'$, which requires some additional definitions
  (see Figure~\ref{fig:hole}).
  The holes of $\tau'$ are made of three types of cells:
  {\em v-cells} which may have been part of a vertex gadget,
  {\em e-cells} which may have been part of an edge gadget,
  and {\em h-cells} the remaning cells.
  For a letter $\text{x} \in \{\text{v},\text{h}\}$ we call x-component
  a connected component of x-cells in $G_{\tau''}$.
  Note that each v-component has 6 cells, and each h-component has 13 cells.
  Given a hole $\tau''$ of $\tau'$, we define the graph $H_{\tau''}$
  whose vertices are the v- and h-components of $G_{\tau''}$,
  and such that two components are adjacent when two of their cells are.
  Remark that $H_{\tau''}$ is a bipartite connected graph,
  since we have not considered e-cells.
  Now we attach each e-cell to a v-component vertex of $H_{\tau''}$
  adjacent to it, according to some map $\vartheta$ from the
  e-cells to the v-components of $G_{\tau''}$.
  %
  %
  For any vertex $v$ of $H$ and the associated vertex gadget,
  we call each of the following couples of cells an {\em access}:
  $(c_1(v),c_6(v))$, $(c_2(v),c_4(v))$, $(c_3(v),c_5(v))$.
  We associate to $\tau''$ a set of {\em adjacent accesses}, which are the
  vertex gadget's access whose cells are adjacent to cells of $\tau''$.
  Among these, the {\em canonical access} of $\tau''$ is the maximal one
  according to some fixed direction ({\em e.g.} the leftmost in our figures).
  Accesses will be used to plug partial paths together,
  and we extend their definition to the v-components of $\tau''$.
  Finally, given $\tau''$ and an adjacent access $\alpha$,
  let $T_{\tau''}^\alpha$ be a rooted spanning tree of $H_{\tau''}$,
  whose root is the h-component of $\tau''$ adjacent to $\alpha$
  (observe that there is a unique such h-component).
  Given a vertex $t$ of $H_{\tau''}$ which is a h-component,
  let $\tau''[t]$ be the subset of $\tau''$ corresponding to
  the subtree of $T_{\tau''}^\alpha$ rooted at $t$.
  %
  We build the {\em Hamiltonian path $P_{\tau''}^\alpha$ along
  $T_{\tau''}^\alpha$}, from one of the access cell of $\alpha$ to the other,
  recursively as follows.
  \begin{itemize}
    \item For the root $r$ of $T_{\tau''}^\alpha$, which is a h-component, 
      build the {\em h-basepath} depicted on the left of Figure~\ref{fig:basepath},
      rotated in order to fit the access $\alpha$.
    \item For each child $s$ of $r$, which is a v-component, build the
      {\em v-basepath} depicted on the right of Figure~\ref{fig:basepath},
      including its associated e-cells $t$ such that $\vartheta(t)=s$.
    \item For each child $t$ of $s$, consider the access $\beta$ of $s$
      to which $t$ is adjacent, and build recursively a Hamitlonian path
      $P_{\tau''[t]}^\beta$ along the subtree $T_{\tau''[t]}^\beta$.
    \item Finally, plug the pieces together at accesses, by operating
      the flips illustrated on Figure~\ref{fig:flips}.
      Remark that by the construction of h-basepath and v-basepath,
      these flips are always possible.
  \end{itemize}
  Now consider, for each hole $\tau''$ of $\tau'$, the path $T_{\tau''}^\alpha$
  with $\alpha$ the canonical access of $\tau''$, and add the adjacencies
  of this path to $p$.

  \begin{figure}
    \centerline{\includegraphics{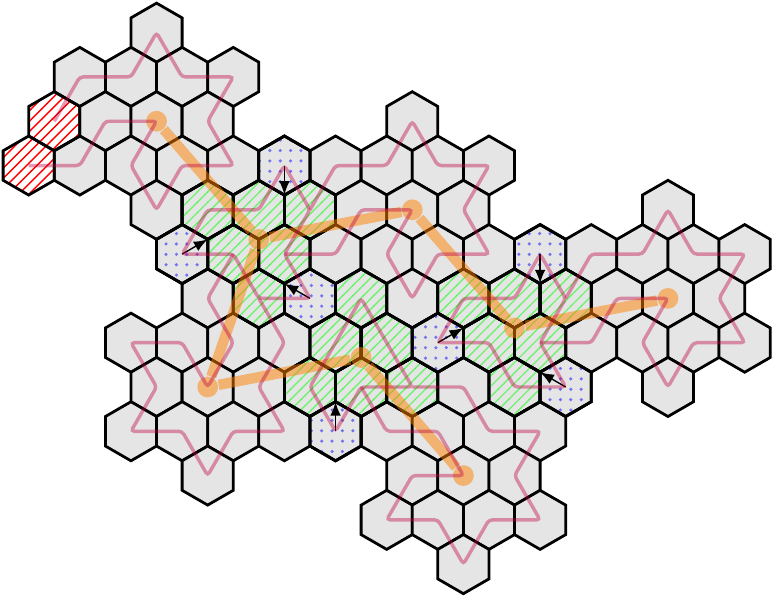}}
    \vspace*{-1cm}
    \centerline{\includegraphics{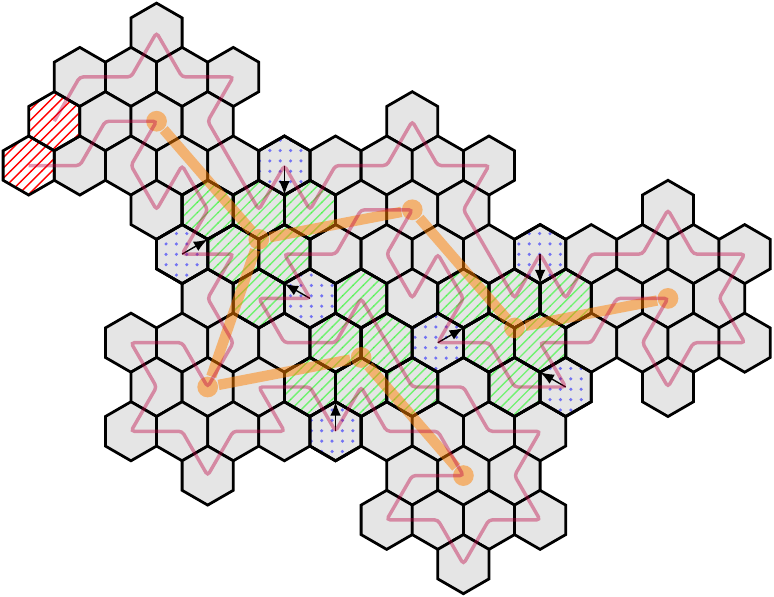}}
    \vspace*{-1cm}
    \centerline{\includegraphics{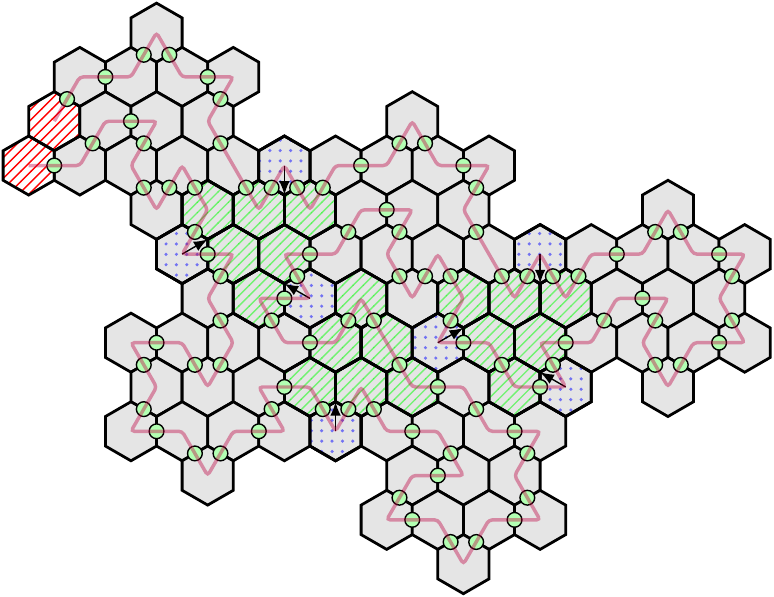}}
    \caption{
      Example construction of $p$ from a Hamiltonian path for one hole $\tau''$.
      The function $\vartheta$ is illustrated on each e-cell $t$
      with an arrow pointing towards the v-component $\vartheta(t)$,
      v- e- and h-cells have slightly marked patterns,
      v- and h-components hold an orange node of $H_{\tau''}$,
      edges of a spanning tree $T_{\tau''}^\alpha$ are drawn in orange
      with the canonical access $\alpha$ hatched in red.
      Top: the h- and v-basepaths before the flips.
      Middle: the Hamiltonian path $P_{\tau''}^\alpha$ obtained after the flips.
      Bottom: the corresponding adjacencies in $p$.
    }
    \label{fig:hole}
  \end{figure}

  \begin{figure}
    \centerline{
      \includegraphics{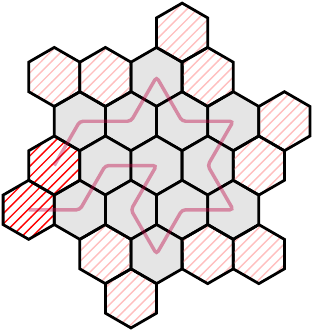}
      \hspace*{2cm}
      \includegraphics{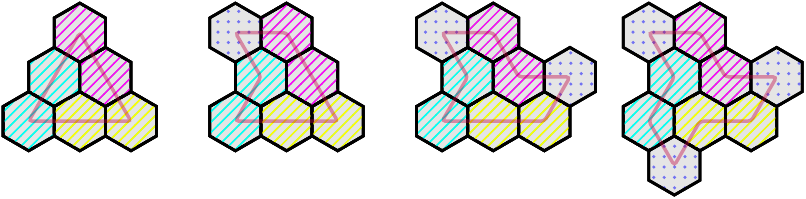}
    }
    \caption{
      Left: in purple the h-basepath according to the (canonical) access hatched in red,
      other accesses are hatched in pink.
      Right: in purple the v-basebaths according to the number of e-cells attached to the
      v-component, the three accesses are hatched in three different colors.
    }
    \label{fig:basepath}
  \end{figure}

  \begin{figure}
    \centerline{\includegraphics{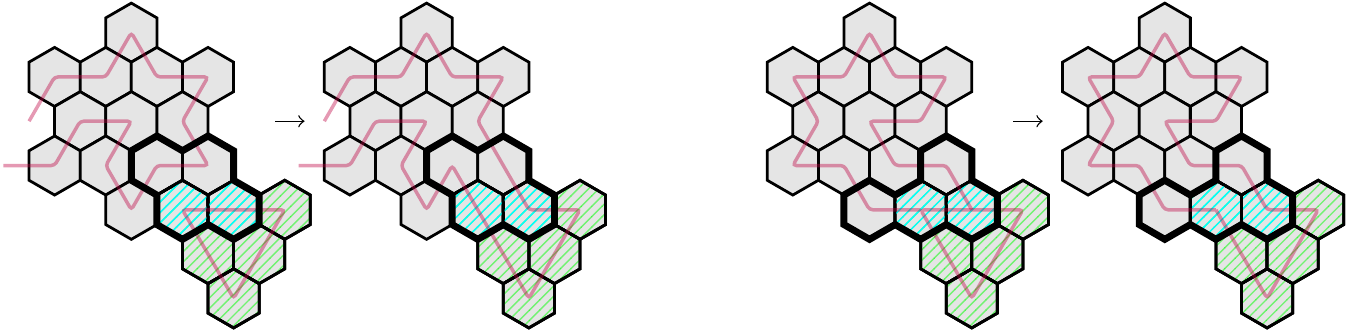}}
    \caption{
      Flips operated in order to connect h-basepath and v-basepath
      (on a downward triangle without e-cell in this example),
      the h-component cells are filled in grey and the v-component access in
      hatched in cyan. The location of the flip is highlighted.
      Left: from a h-basepath parent to a v-basepath child.
      Right: from a v-basepath parent to a h-basepath child.
    }
    \label{fig:flips}
  \end{figure}

  For $m$, choose an access $(t,t')$ of some vertex gadget corresponding
  to a vertex of $H$ of degree two, such that
  $t,t'$ do not appear in $p$, and set $m(t)=1$, $m(t')=n$
  with $n$ the total number of cells in $\tau$.

  \medskip

  If a cell belongs to two elements of $p$, {\em i.e.} has two enforced adjacencies,
  then any Hamiltonian path from $1$ to $n$ crosses this cell along these adjacencies.
  Thus, as a consequence of the definition of $p$,
  the game $(\tau,m',p)$ is equivalent (in terms of decision) to the game
  $(\tau',m',p')$,
  where $m'$ is the same as $m$ except that $m(t')=n'$ with $n'$ the total number of cells in $\tau'$, and
  where $p'$ are the couples of canonical access cells in
  vertex gadgets for each hole of $\tau'$ (see Figure~\ref{fig:withoutholes-prime}).

  \begin{figure}
    \centerline{
      \includegraphics{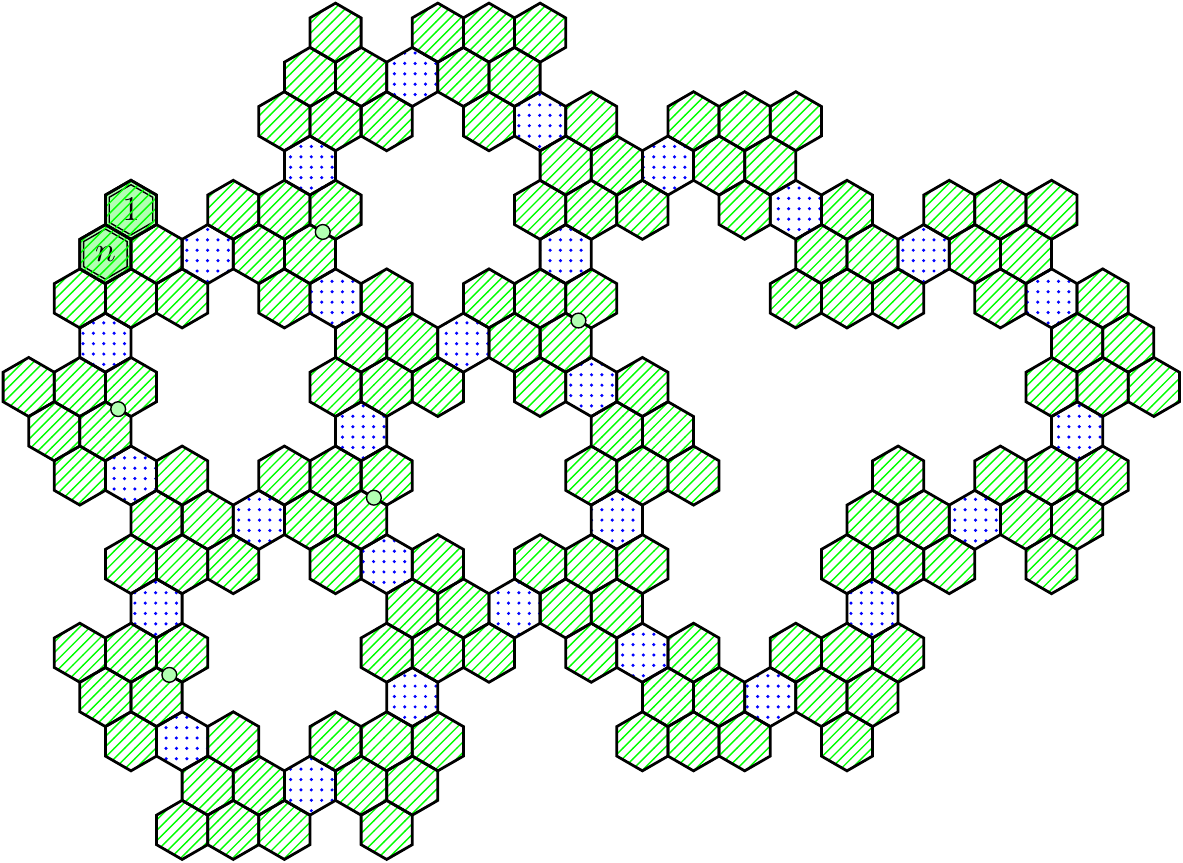}
    }
    \caption{
      Game $(\tau',m',p')$ equivalent to $(\tau,m,p)$ from Figure~\ref{fig:withoutholes}.
    }
    \label{fig:withoutholes-prime}
    \vspace*{.5cm}
    \centerline{
      \includegraphics{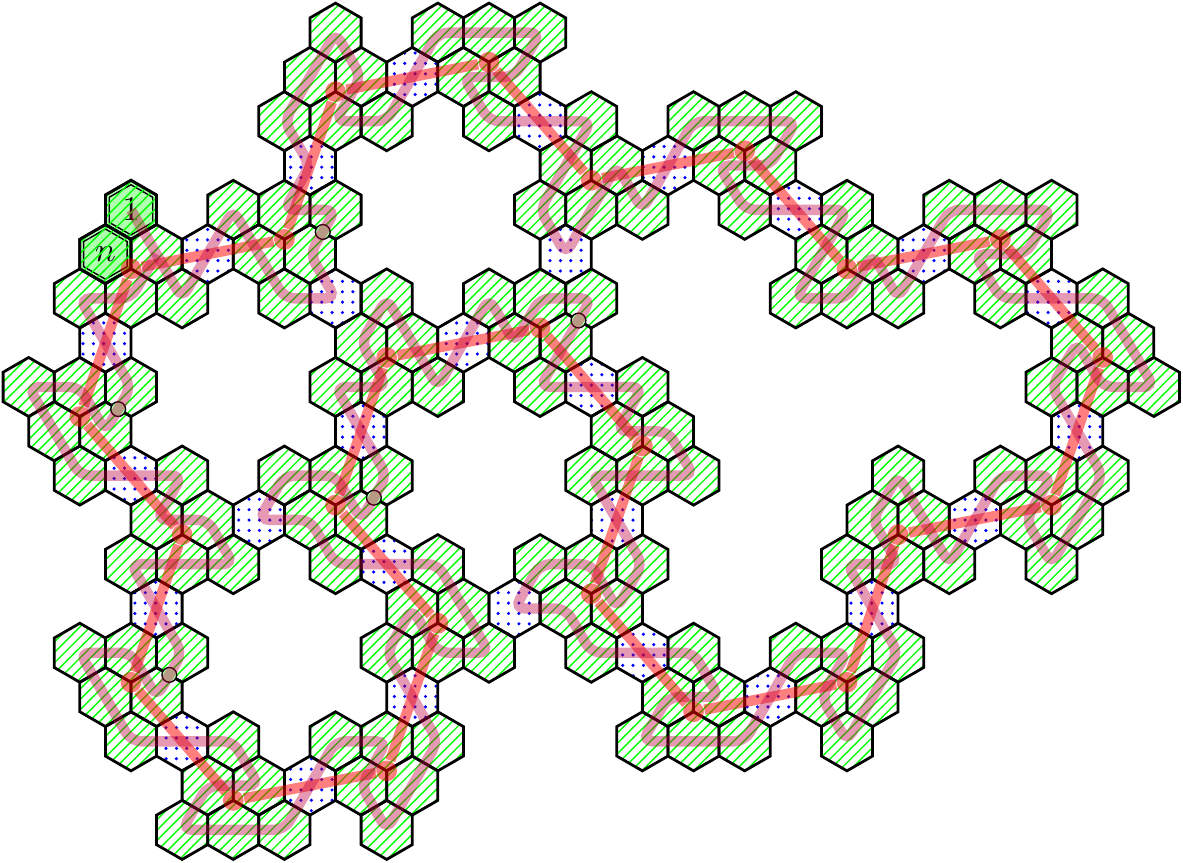}
    }
    \caption{
      A solution to the game $(\tau',m',p')$
      with only the Hamiltonian path from $1$ to $n$ depicted in purple,
      and the corresponding Hamiltonian cycle on $G$ in red.
    }
    \label{fig:withoutholes-prime-solution}
  \end{figure}

  \medskip

  Now let us argue that $H$ has a Hamiltonian cycle if and only if the game $(\tau',m',p')$ has a solution.
  If $H$ has a Hamiltonian cycle, then it is straightforward to construct a solution
  to the game $(\tau',m',p')$ respecting the adjacencies given by $p$,
  as shown on Figure~\ref{fig:withoutholes-prime-solution},
  by starting from the vertex gadget hosting numbers $1,n$ and following the Hamiltonian cycle on $H$:
  for each vertex gadget, construct a Hamiltonian path from on edge gadget to the other,
  and include the third edge gadget to the path if it has not been included so far
  (recall that the vertices of $H$ have degree two or three);
  the vertex gadget hosting numbers $1,n$ has a special pattern
  (it always corresponds to a vertex of $H$ of degree two).
  Note that the use of canonical accesses ensures that each vertex gadget has
  at most one access in $p'$.
  All the possibilities are presented on Figure~\ref{fig:wihtoutholes-solutions}.

  \begin{figure}
    \centerline{\includegraphics{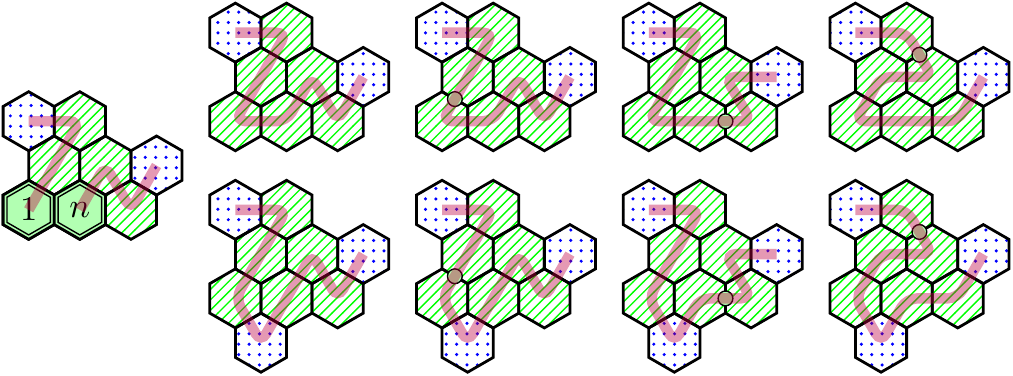}}
    \caption{
      Hamiltonian path for all the possibilities of vertex and edge gadgets combination
      (up to rotation), in order to build a solution to the instance $(\tau',m',p')$
      from a Hamiltonian cycle on $H$.
    }
    \label{fig:wihtoutholes-solutions}
  \end{figure}

  \medskip

  If the game $(\tau',m',p')$ has a solution $s$,
  then the order of vertex gadgets induced by $\adjs{s}$ gives a Hamiltonian path in $H$:
  \begin{itemize}
    \item it is only possible to go from one vertex gadget to another vertex gadget
      by going over an edge gadget, hence adjacencies of $H$ are respected,
    \item an edge gadget consists of only one cell, thus we cannot use twice an edge of $H$, and the degree of every vertex in $H$ is two or three; it follows that we cannot visit twice a vertex of $H$.
  \end{itemize}
  Since a solution to the game includes all cells ({\em i.e.} all vertex gadgets)
  and goes back to the starting vertex gadget thanks to the positions of $1$ and $n$
  given by $m'$, we conclude that the corresponding cycle on $H$ is Hamiltonian.
\end{proof}

\begin{remark}
  It is a result of~\cite{rz00} that a h-component
  corresponds to the only finite, 2-connected, linearly convex,
  subgraph of the Archimedean triangular tiling
  (the graph obtained from the hexagonal grid cells with adajcencies $\adj$)
  which do not have a Hamiltonian cycle;
  but in our construction we only need Hamiltonian paths on h-components.
\end{remark}

\begin{remark}
  With almost trivial adaptations of the proof of Theorem~\ref{theorem:withoutholes},
  one can also obtain the $\NP$-hardness when the shape of the game is a hexagon
  (as on the left of Figure~\ref{fig:rikudo}),
  and with one missing cell at the center, such as in ``real'' rikudo games.
\end{remark}

\section{Rikudo with already placed numbers}
\label{s:conditions}

The hardness proofs presented in Section~\ref{s:hardness}
reduce from Hamiltonian cycle problems,
and we (almost) do not use the already placed numbers given by $m$
(intuitively, because it would have required to have some prior knowledge on the path).
It is therefore natural to ask whether {\em Rikudo} games may become easier to solve when it is imposed
that some numbers are initially placed by $m$?
Given that the reduction leading to Theorem~\ref{theorem:withoutholes}
let $m$ have domain $\dom(m)=\{1,n\}$, we straightforwardly have that the following variant is still $\NP$-hard.
\\[.5em]
\decisionpb{$\alpha$-Rikudo without holes {\normalfont (for some real constant $0 < \alpha \leq 1$)}}
{a game $(\tau,m,p)$ such that $|\dom(m)| \geq \alpha n$ with $n=|\tau|$.}
{does it admit a solution?}\\[.5em]
Indeed, one can construct a game for some $n'$ and then do some padding,
by placing with $m$ all the numbers from $n'+1$ until $n=\lceil\frac{n'}{\alpha}\rceil$.

The following question is then of particular interest.\\[.5em]
\decisionpb{1-over-$k$-Rikudo {\normalfont (for some integer constant $k\geq 2$)}}
{a game $(\tau,m,p)$ such that $\dom(m) = \{xk+1 \mid x\in\N\} \cap [n]$ with $n=|\tau|$.}
{does it admit a solution?}\\[.5em]
In words, it imposes that $m$ places the numbers $1,k+1,2k+1,3k+1,\dots$ {\em i.e.} one number every $k$ numbers.
The total fraction $\frac{|\dom(m)|}{n}$ approaches $\frac{1}{k}$,
but contrary to {\bf $\alpha$-Rikudo} the repartition of numbers is constrained.
In the case $k=2$ all odd numbers are already placed and the problem turn out to be in $\Poly$
(Theorem~\ref{theorem:1-over-2}), whereas we prove that it is $\NP$-hard for $k=3$
(Theorem~\ref{theorem:1-over-3}).

\begin{theorem}
  \label{theorem:1-over-2}
  {\bf 1-over-2-Rikudo} is in $\Poly$.
\end{theorem}

\begin{proof}
  The holes do not matter in this simple reduction to {\bf 2-SAT}.
  Given a game $(\tau,m,p)$ with all odd numbers placed,
  it is an easy observation that for any pair of already placed integers $x,x+2$
  there are at most two possible positions for the number $x+1$.
  The reduction goes as follows: start by placing all (even) numbers having a unique possible position
  until no such number exist (it corresponds to performing unit propagation).
  Then for all remaining (even) numbers $x$, create two variables $v_x^1$ and $v_x^2$
  corresponding to the two possible positions.
  Construct a formula having the following clauses:
  \begin{itemize}
    \item $v_x^1 \vee v_x^2$ for each remaining (even) number $x$,
    \item $\neg v_x^i \vee \neg v_y^j$ for each pair of variables corresponding to the same position.
  \end{itemize}
  The variable creation ensure that the numbers form a path,
  the first set of clauses ensure that all remaining numbers are placed,
  and the second set of clauses ensure that no two numbers are placed at the same position.
\end{proof}

\begin{theorem}
  \label{theorem:1-over-3}
  {\bf 1-over-3-Rikudo} is $\NP$-hard.
\end{theorem}

\begin{proof}
  We present a reduction from the problem {\bf Planar 1-in-3-SAT}
  (clauses of size three, satisfied by exactly one literal),
  which is proven to be $\NP$-hard in~\cite{df86}.
  A formula $\phi$ in conjunctive normal form (CNF) is {\em planar} when so is the bipartite graph $G_\phi$
  having one vertex for each variable of $\phi$, one vertex for each clause of $\phi$,
  and an edge between a variable $x_i$ and a clause $c_j$ whenever $x_i$ appears in $c_j$.
  From a planar 3-CNF formula $\phi$ (checking planarity can be done in linear time~\cite{ht74}),
  we slightly modify the graph $G_\phi$ while preserving planarity:
  each variable vertex $x_i$ is replaced by a binary tree (called {\em variable tree}) having as many leaves as
  occurrence of $x_i$ in $\phi$, and each of these leaves is connected to one clause in which $x_i$ appears.
  We also add a negation vertex between variable $x_i$ and clause $c_j$ if $x_j$ appears as $\neg x_i$ in $c_j$.
  We consider a planar embedding of this new graph $G_\phi$ into the graph underlying a flat hexagonal grid
  ({\em i.e.} vertices of $G_\phi$ are cells of the grid, and edges follow cell to cell adjacencies).
  Such an embedding can easily be computed in polynomial time
  (see~\cite{k93,t74} for reference, but naive greedy methods are enough for our purpose).
  Finally, for technical reasons to be explained later in this proof,
  we scale up the obtained graph on the hexagonal grid by a factor of two.
  See Figure~\ref{fig:planar1in3sat-graph} for an example.

  \begin{figure}
    \centerline{\includegraphics[scale=.5]{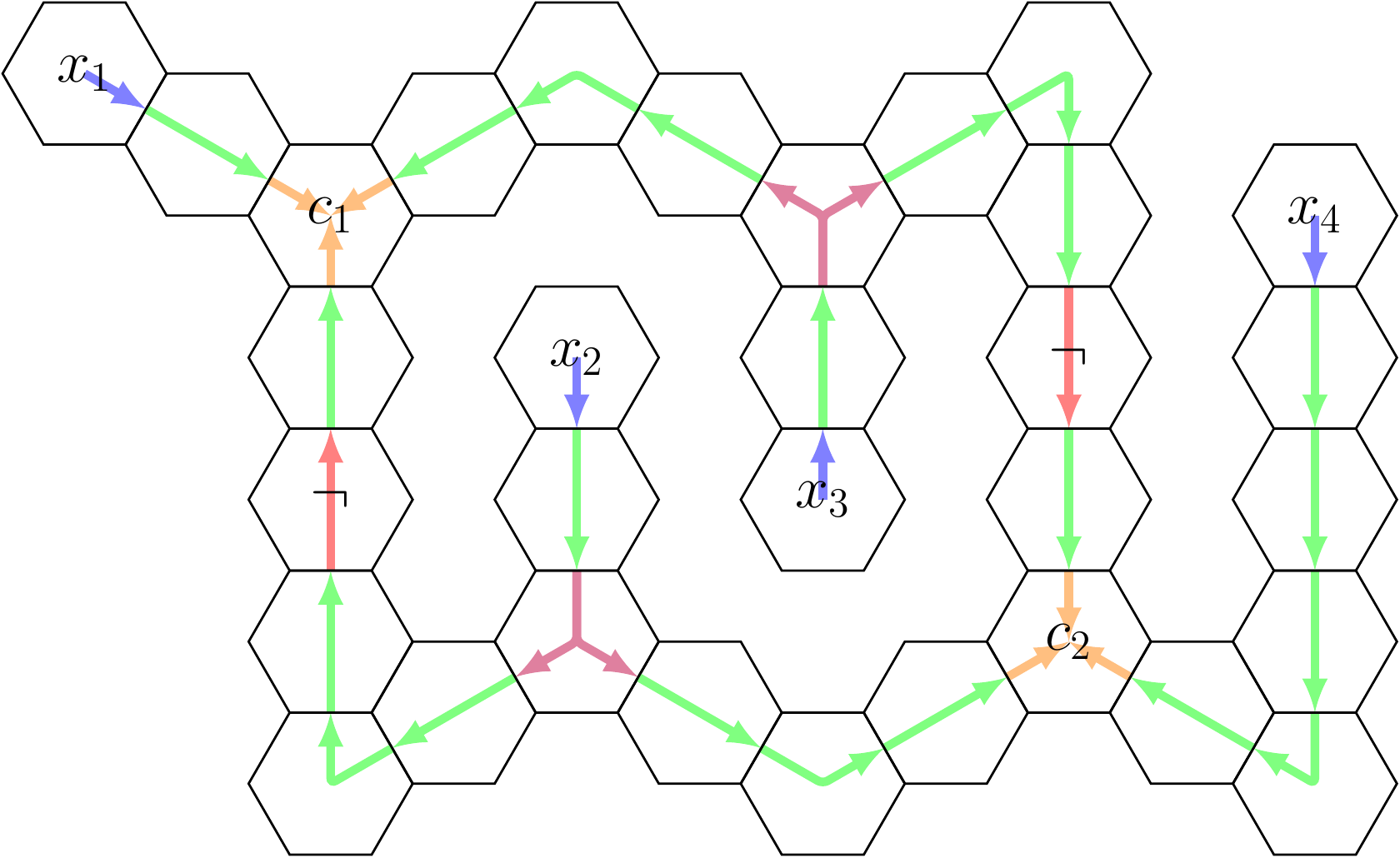}}
    \caption{
      Example embedding of the planar formula $\phi=(x_1 \vee \neg x_2 \vee x_3) \wedge (x_2 \vee \neg x_3 \vee x_4)$
      and its graph $G_\phi$ into a flat hexagonal grid.
      Variables in blue (choice-cells),
      binary trees in purple (duplicate-cells),
      negations in red (negation-cells),
      clauses in orange (clause-cells),
      edges in green (wire-cells).
    }
    \label{fig:planar1in3sat-graph}
  \end{figure}

  \medskip
  
  A {\em Rikudo} game is obtained by replacing each cell of the embedding of
  $G_\phi$ by a macrocell for the game, depending on the content of this cell.
  We distinguish five types
  of cells, and five corresponding types of macrocells.
  The game {\em macrocells} are formally defined on a subset of cells forming a flat hexagon of side length $9$,
  which are assembled side by side to form a {\em Rikudo} game.
  On the sides one row of cells from adjacent macrocells are merged.
  Illustrations are presented on Figure~\ref{fig:planar1in3sat-macrocells}.

  \begin{remark*}
    When we refer to a {\em solution} of a macrocell,
    we mean to place all the numbers between the endpoints of its subpath(s) of numbers
    (macrocells have one, two or three subpaths). Furthermore, a solution {\em must}
    place numbers on all game cells {\em within} the macrocell
    ({\em i.e.} except possibly on its sides, precisions are given in the next remark on input bit).
    Indeed, if this is not the case then cells left empty within a macrocell will be left empty in
    the solution obtained by the assembly of macrocells' solutions.
  \end{remark*}

  \begin{itemize}
    \item The root of each variable tree is called a {\em choice-cell},
      its macrocell admits two kinds of solutions
      corresponding to {\em true} and {\em false} (one {\em bit}).
      In each solution one of the two cells on the top side hosts a number, and the other one does not:
      in the {\em true} solutions the left cell contains a number and the right cell does not,
      whereas in the {\em false} solutions the left cell does not contain a number and the right cell does.
      Remark that this is the convention for an {\em output} bit, for an {\em input} bit it is reversed.
  \end{itemize}

  \begin{remark*}
    More generally, the input and output bits of a macrocell are defined as follows.
    Let us call {\em number-cell} a cell where a number is already placed ({\em i.e.} in $\dom(m)$),
    and {\em game-cell} a cell where the player will place a number ({\em i.e.} in $\tau\setminus\dom(m)$).

    Consider the row of cells on an {\em output side}.
    In the clockwise order, it has: one number-cell, one game-cell $A$,
    one number-cell, and one game-cell $B$.
    A {\em true} (respectively {\em false}) {\em output bit} corresponds to cell $A$
    (respectively $B$) containing a number from this macrocell,
    whereas cell $B$ (respectively $A$) does not.

    Symmetrically, consider the row of cells on an {\em input side}.
    In the clockwise order, it has: one game-cell $B$, one number-cell,
    one game-cell $A$, and one number-cell.
    A {\em true} (respectively {\em false}) {\em input bit} corresponds to cell $A$
    (respectively $B$) containing a number from the adjacent macrocell,
    whereas cell $B$ (respectively $A$) does not. As a consequence, a solution to this macrocell
    must place a number in cell $B$ (respectively $A$), but not in cell $A$ (respectively $B$).

    Also observe that the merge of an input side with an output side matches the position of game-cells
    and numbers-cells. The assembly of macrocells will be precised just after the presentation of all macrocells.
  \end{remark*}

  \begin{itemize}
    \item An internal node of a variable tree is called a {\em duplicate-cell},
      given an input bit on the bottom side its macrocell admits only solutions which copy this bit
      to the top left and top right sides (it has one input and two outputs).
      In a solution each variable tree therefore has the same bit on all its leaves.
    \item A negation vertex is called a {\em negation-cell},
      given an input bit on the bottom side its macrocell admits only solutions which flip this bit
      to the top side (it has one input and one output).
    \item A clause vertex is called a {\em clause-cell},
      it has a solution if and only if exactly one of the three input bits on its sides is {\em true}
      (it has three inputs).
    \item All other non-empty cells are called {\em wire-cells},
      given an input bit on one side the macrocells admit only solutions
      which copy this bit to the other side (it has one input and one output).
  \end{itemize}
  These claims can easily be verified by hand on the figures from Appendix~\ref{a:macrocells},
  for the interested reader.
  Macrocells can be rotated (but not reflected, since this corrupts the subpaths merging).
  Note that a bit of information reads differently as an input and as an output of a macrocell,
  and that the subpath endpoints are placed accordingly so that an output merges an input.
  Macrocells are directed as the graph $G_\phi$ from Figure~\ref{fig:planar1in3sat-graph}.

  \begin{figure}
    \centering
    \begin{tikzpicture}
      \node (choice) at (-4.5,2.8) {\includegraphics[scale=.65]{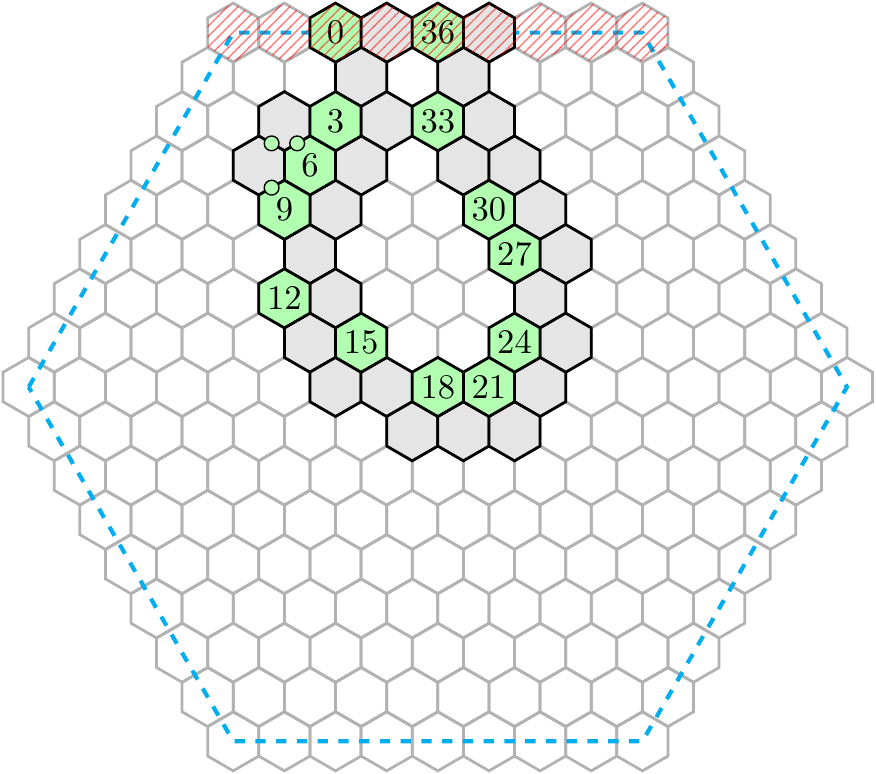}};
      \node at (choice.north) {\small Choice-macrocell};
      \node (duplicate) at (4.5,2.8) {\includegraphics[scale=.65]{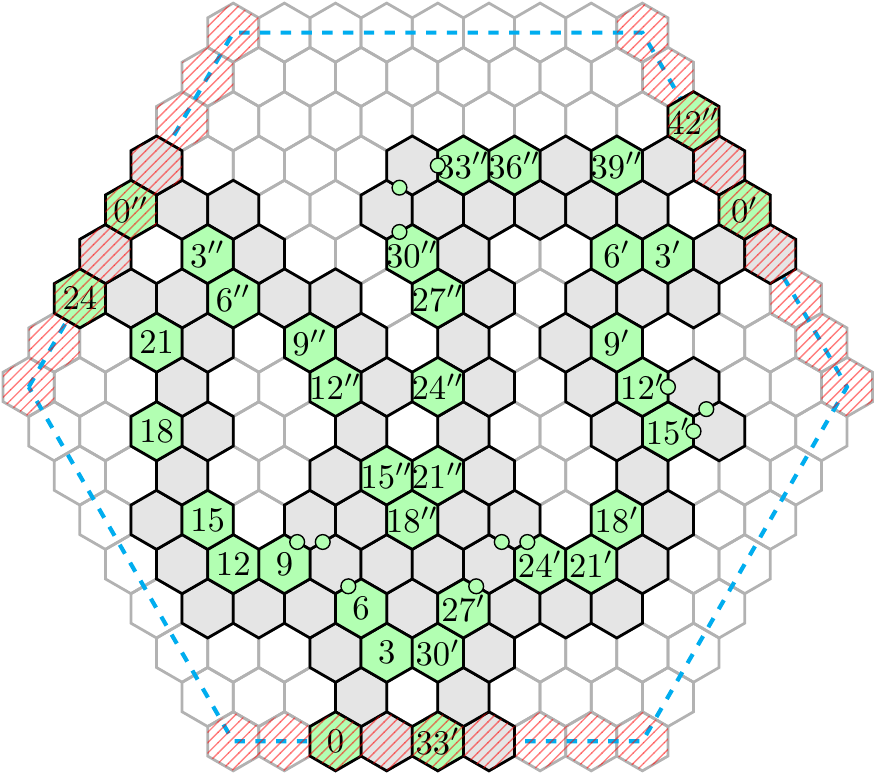}};
      \node at (duplicate.north) {\small Duplicate-macrocell};
      \node (negation) at (0,0) {\includegraphics[scale=.65]{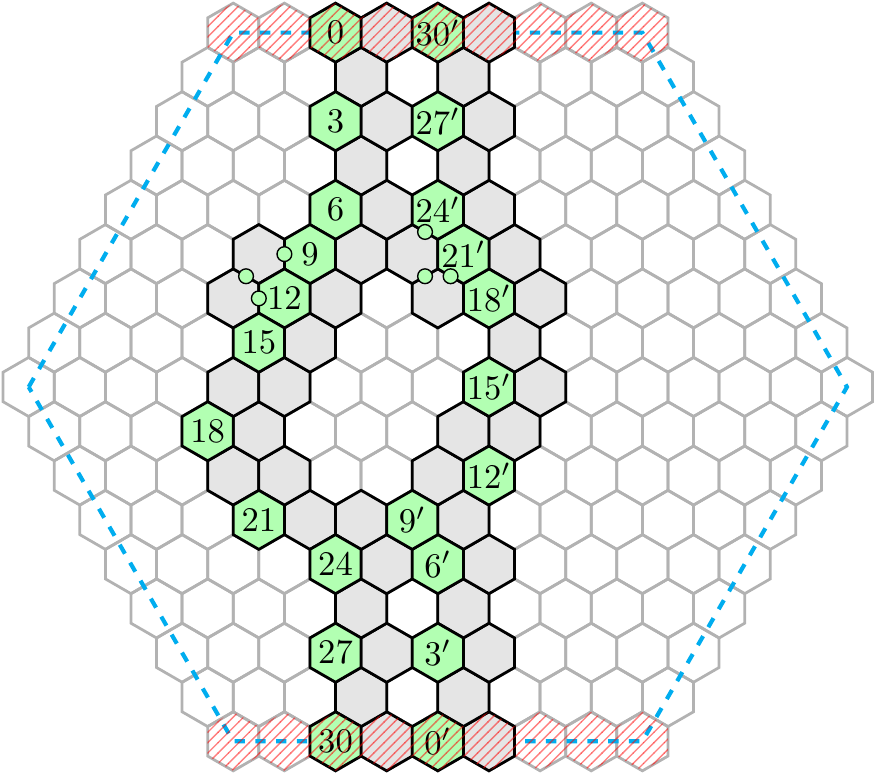}};
      \node at (negation.north) {\small Negation-macrocell};
      \node (clause) at (-4.5,-2.8) {\includegraphics[scale=.65]{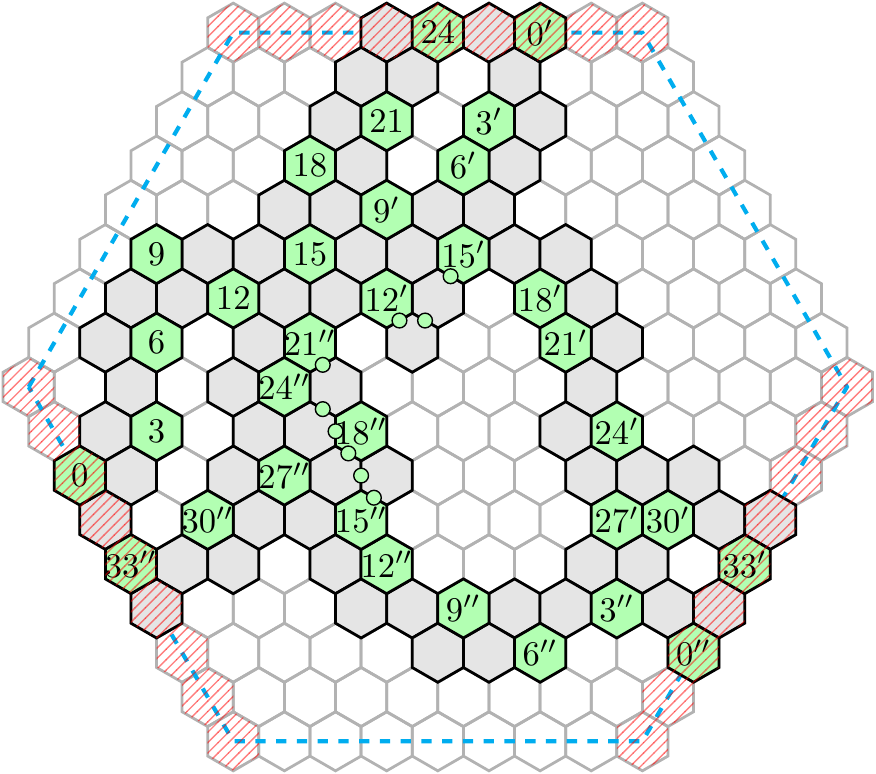}};
      \node at (clause.north) {\small Clause-macrocell};
      \node (wire) at (4.5,-2.8) {\includegraphics[scale=.65]{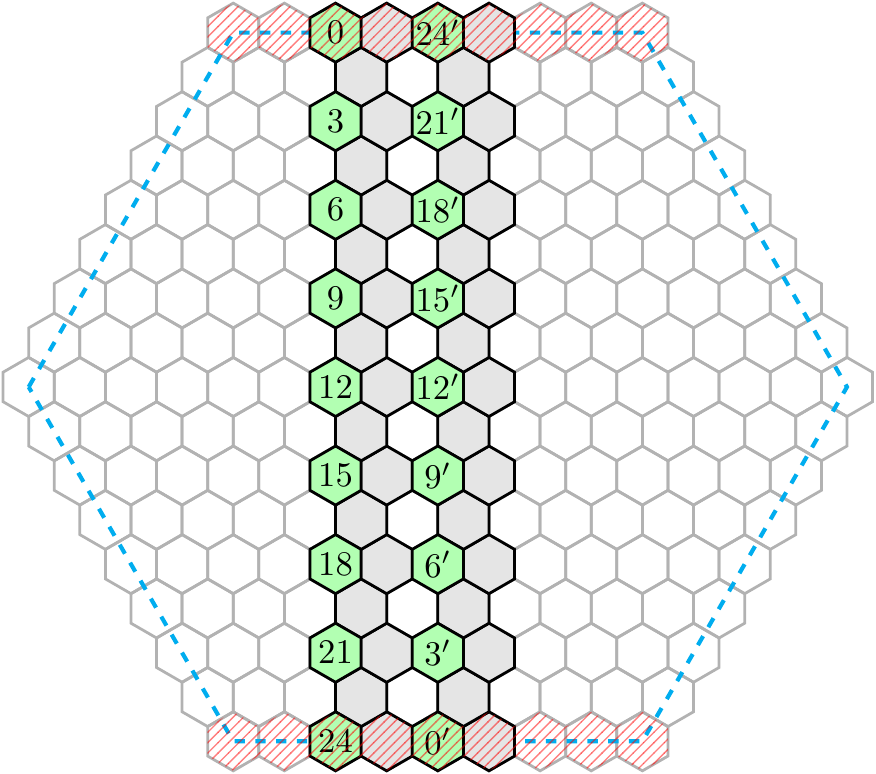}};
      \node at (wire.north) {\small Wire-macrocell};
      \node (wires) at (0,-8.5) {\includegraphics[scale=.5]{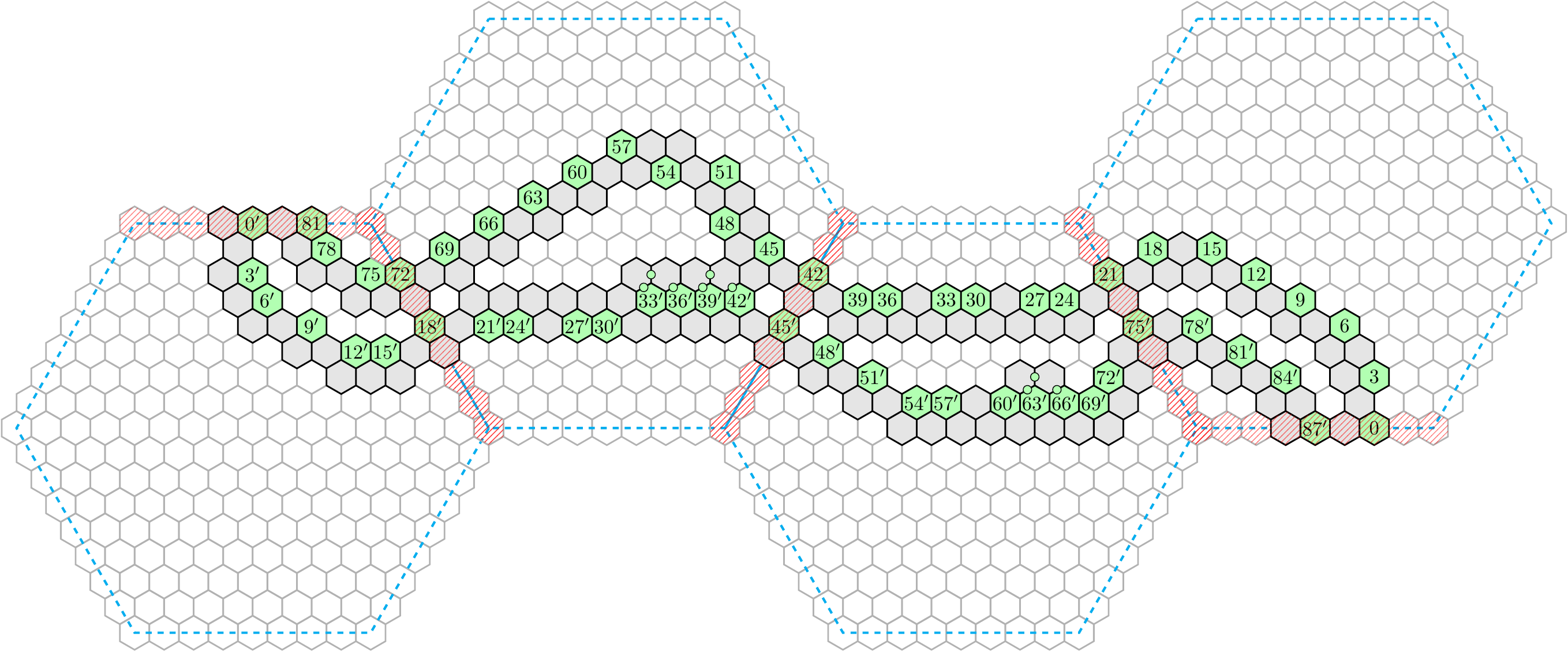}};
      \node at (wires.north) {\small Wire-macrocells};
    \end{tikzpicture}
    \caption{
      Macrocell types,
      with the rows merged with the adjacent macrocell hashed in red (input/output).
      Distinct subpaths of numbers are marked with $x$, $x'$ and $x''$
      starting at $0$, $0'$ and $0''$ respectively.
      Choice-macrocell (output on the top),
      duplicate-macrocell (one input on the bottom and two outputs),
      negation-macrocell (input on the bottom and output on top),
      clause-macrocell (three inputs),
      straight wire-macrocell (input on the bottom and output on top),
      and four bending wire-macrocells assembled
      (from input on left to output on right).
    }
    \label{fig:planar1in3sat-macrocells}
  \end{figure}

  \medskip

  Let us now give precisions on how macrocells are assembled,
  so that the obtained game may have a solution forming only one long path from $1$ to $n$
  (made by connecting all macrocell's subpaths together).
  One bit of information is transported by two subpaths of numbers.
  The game path will simply follow the Eulerian path on the doubling tree\footnote{The
  doubling tree consists in replacing each edge of the tree by two copies of itself,
  hence the Eulerian path corresponds to walking along the contour of the spanning tree.}
  associated to a spanning tree of $G_\phi$.
  The numbers already placed on two adjacent macrocells are identified
  at the merged row of cells on the common side:
  we can shift the numbers on each subpath of one macrocell to match the other,
  and reverse some subpath of numbers to be coherent relative to the increasing/decreasing order
  (in the walk on the spanning tree the orders of the subpaths of a bit are typically opposite).
  One last macrocell is required for wire-cells corresponding to edges not part of the spanning tree:
  a {\em cutter-macrocell} replaces one of the straight wire-cell on these edges
  (hence the factor two scaling creating a straight wire-cell on every edge).
  See the left of Figure~\ref{fig:planar1in3sat-cutter}, the cutter-macrocell copies the bit from the bottom side
  to the top side, but the two subpaths are shunted differently.
  Finally we choose some start (number $1$) and stop (number $n$) on a remaining straight wire-cell,
  as presented in the {\em start-macrocell} on the right of Figure~\ref{fig:planar1in3sat-start}.
  Figure~\ref{fig:planar1in3sat-path} presents a spanning tree
  and the way subpaths of already placed numbers are assembled.

  \begin{figure}
    \centerline{
      \includegraphics[scale=.65]{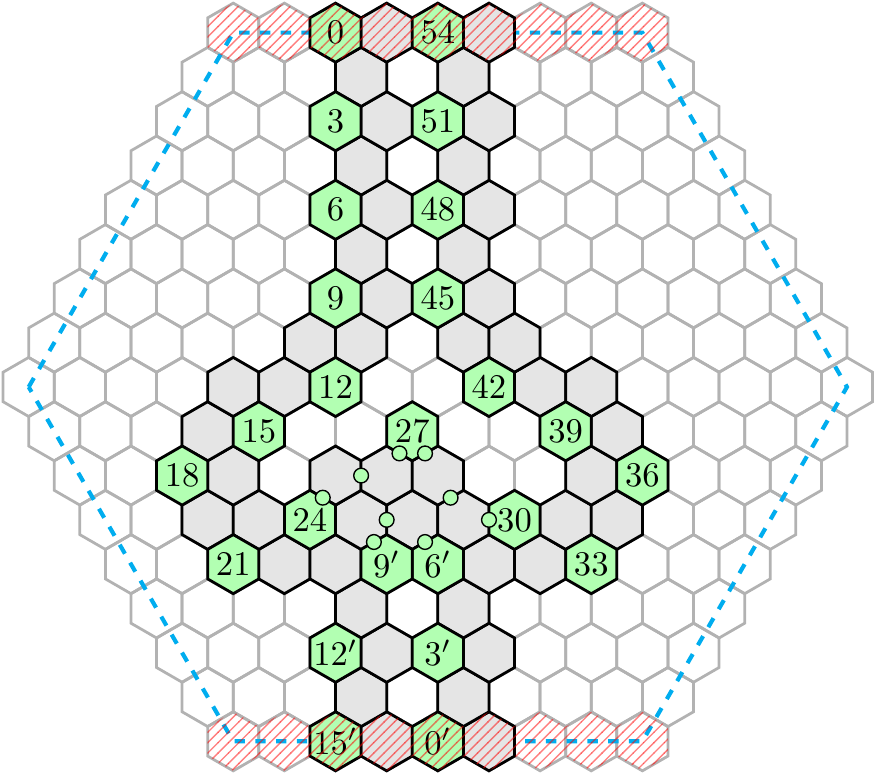}
      \includegraphics[scale=.65]{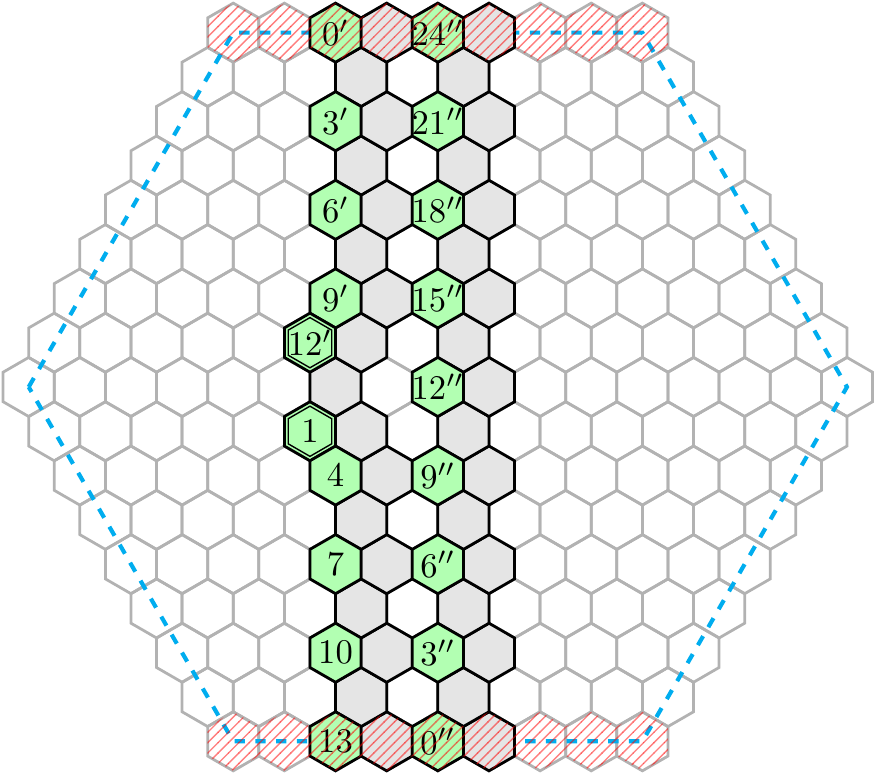}
    }
    \caption{
      Cutter-macrocell (left, input on the bottom and output on top)
      and start-macrocell (right, input on the bottom and output on top,
      the start is $1$ and stop is $12'$, the three subpaths may be shifted independently),
      replacing straight wire-cells.
    }
    \label{fig:planar1in3sat-cutter}
    \label{fig:planar1in3sat-start}
    \vspace*{1em}
    \centerline{\includegraphics[scale=.7]{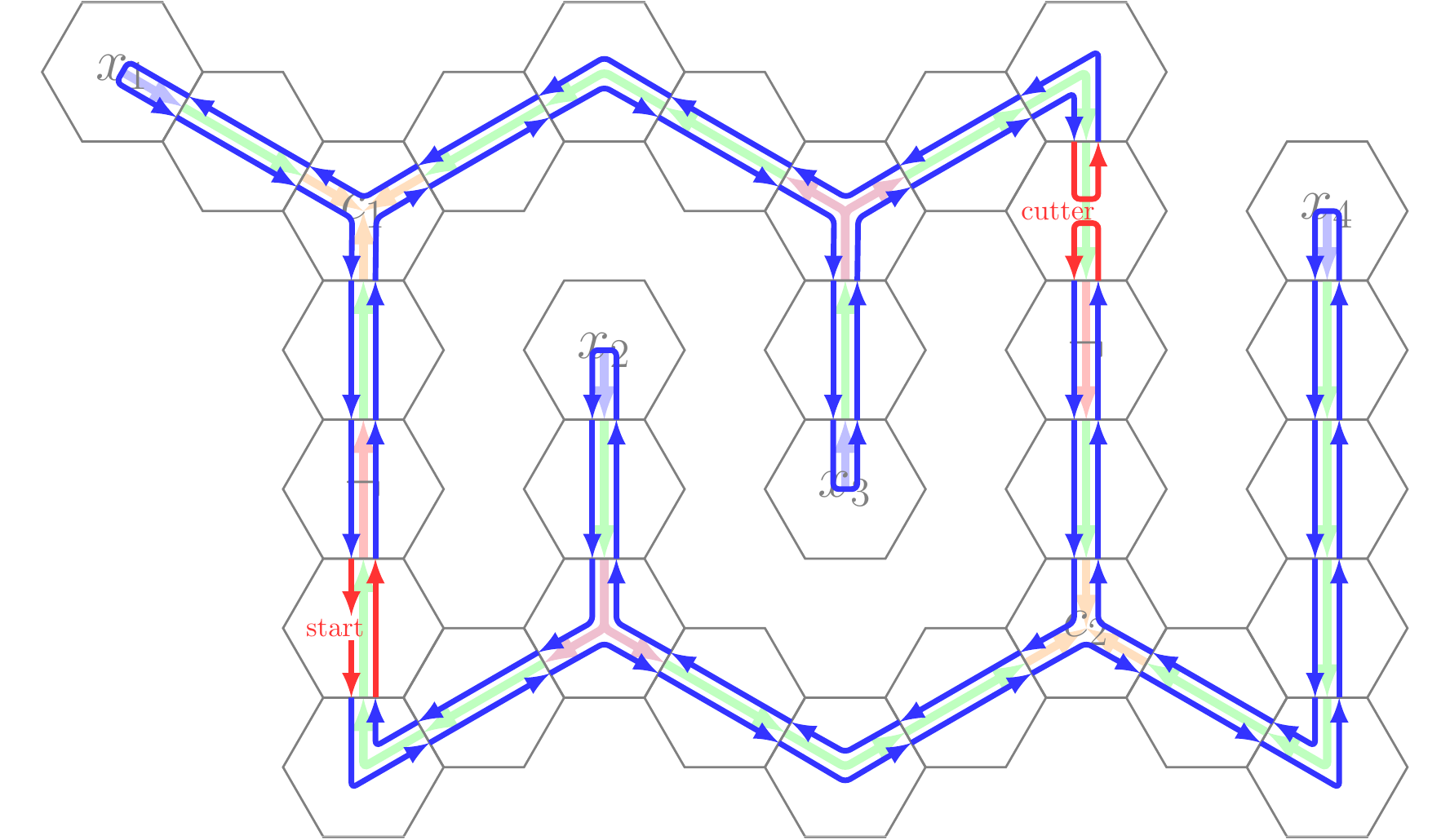}}
    \caption{
      Spanning tree and the way subpaths of already placed numbers are assembled,
      for the graph $G_\phi$ from Figure~\ref{fig:planar1in3sat-graph}.
      Blue arrows follow the increasing order of numbers, with the start-macrocell
      and cutter-macrocells (there is only one in this example) highlighted in red.
      Starting from number $1$ on the start-macrocell, we follow the blue arrows to assign
      subsequent number sequences to the macrocells' subpaths, leading to a unique path from $1$ to $n$
      (if all macrocells admit a solution).
    }
    \label{fig:planar1in3sat-path}
  \end{figure}


  \medskip

  The construction is finished. If $\phi$ has a solution then we choose a satisfying assignment
  for the bits of the choice-macrocells, which are copied by duplicate-macrocells and transported
  by wire-cells to the clause-macrocells, each of these later having a solution
  (and the whole play forms a path from $1$ to $n$ as detailed in the previous paragraph).
  If $\phi$ has no solution then for any assignment for the bits of the choice-macrocells,
  these bits are copied by duplicate-macrocells and transported by wire-macrocells to the
  clause-macrocells, and at least one of the clause-macrocells does not have exactly one
  {\em true} among the three bits on its input sides, hence the game has no solution.
\end{proof}

We strongly believe that {\em 1-over-$k$-Rikudo} is also $\NP$-hard for any $k\geq 4$
based on analogous constructions.
Also remark that the double-wire logics from the proof of Theorem~\ref{theorem:1-over-3}
may allow to simulate arbitrary computation in {\em Rikudo} games
(we have negation-macrocells to perform cross-over,
and clause-macrocells use some disjunction and conjunction of bits).
The game has some more variants ({\em e.g.} with indications in some cells of the parity of the number to be placed),
and it would be interesting to study how these new features may or may not embed complexity,
{\em i.e.} in some sense whether more clues make the game easier or not.


\section*{Acknowledgments}

The authors are thankful to Papicri for bringing those games almost daily during France COVID-19 first lockdown.

\bibliographystyle{plain}
\bibliography{biblio}

\newpage
\appendix

\section{Solutions}
\label{a:solutions}

\begin{figure}[h!]
  \centerline{
    \includegraphics{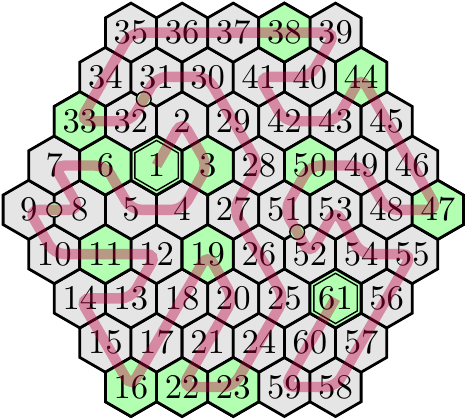}
    \hspace*{.5cm}
    \includegraphics{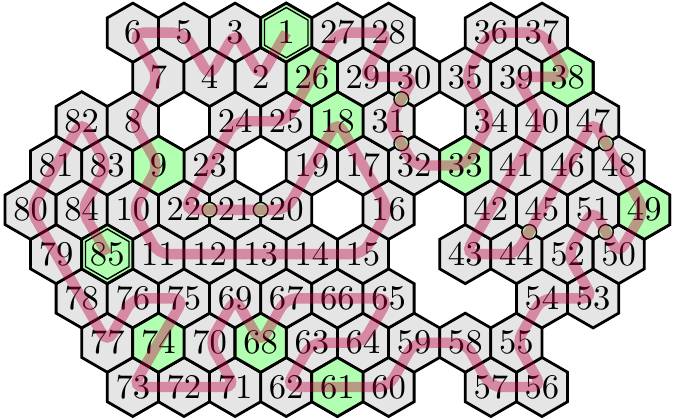}
  }
  \caption{
    Solutions to the games from Figure~\ref{fig:rikudo}.
  }
  \label{fig:solutions}
\end{figure}

\section{Macrocells from the proof of Theorem~\ref{theorem:1-over-3}}
\label{a:macrocells}

Precisions on macrocell solutions,
and on the input and outputs bits, 
are given in the two remarks within the proof of Theorem~\ref{theorem:1-over-3}.

\begin{figure}[h!]
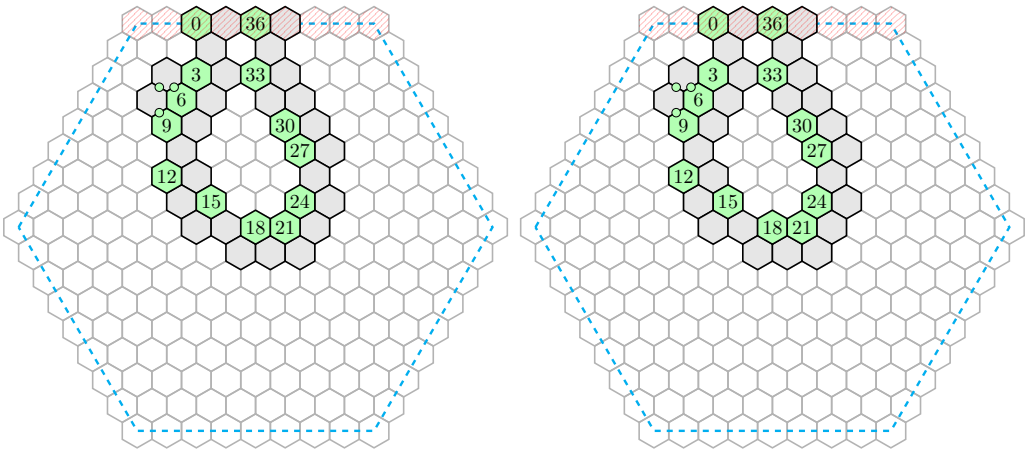

  \centerline{
    \includegraphics[scale=.75]{fig-macrocell-choice.pdf}
    \includegraphics[scale=.75]{fig-macrocell-choice.pdf}
  }
  \caption{
    Choice-macrocells have two kinds of solutions when placing the numbers from $0$ to $36$.
    On the top side, if the left cell hosts a number and the right cell does not then the output bit is {\em true},
    and if the left cell does not host a number and the right cell does then the output bit is {\em false}.
  }
  \label{fig:macrocell-choice}
\end{figure}

\begin{figure}[h!]
  \centerline{
    \includegraphics[scale=.75]{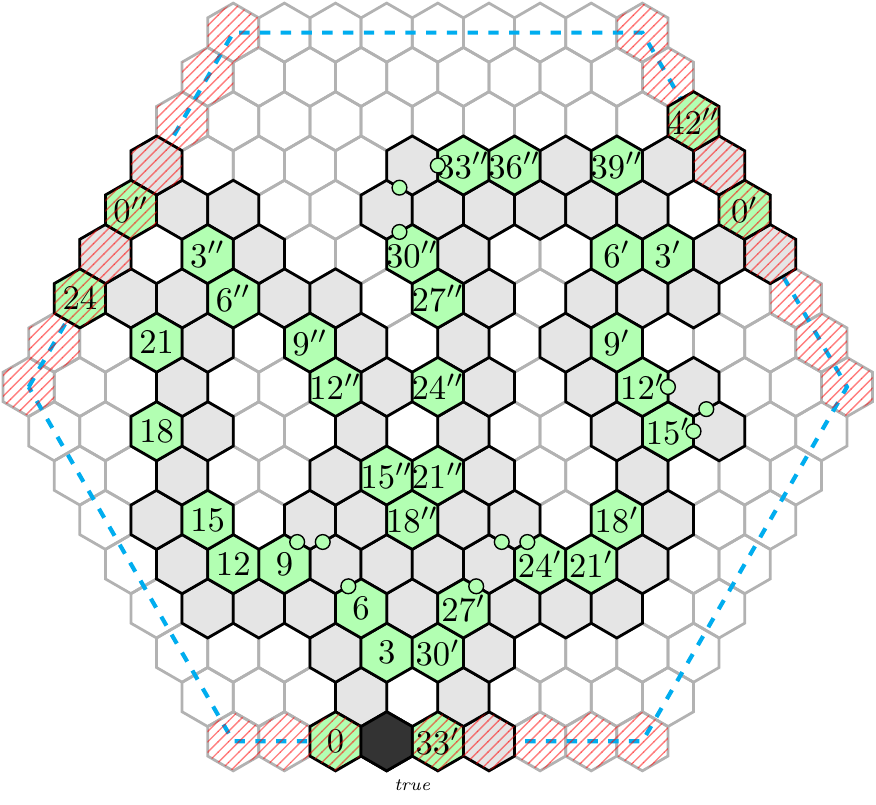}
    \includegraphics[scale=.75]{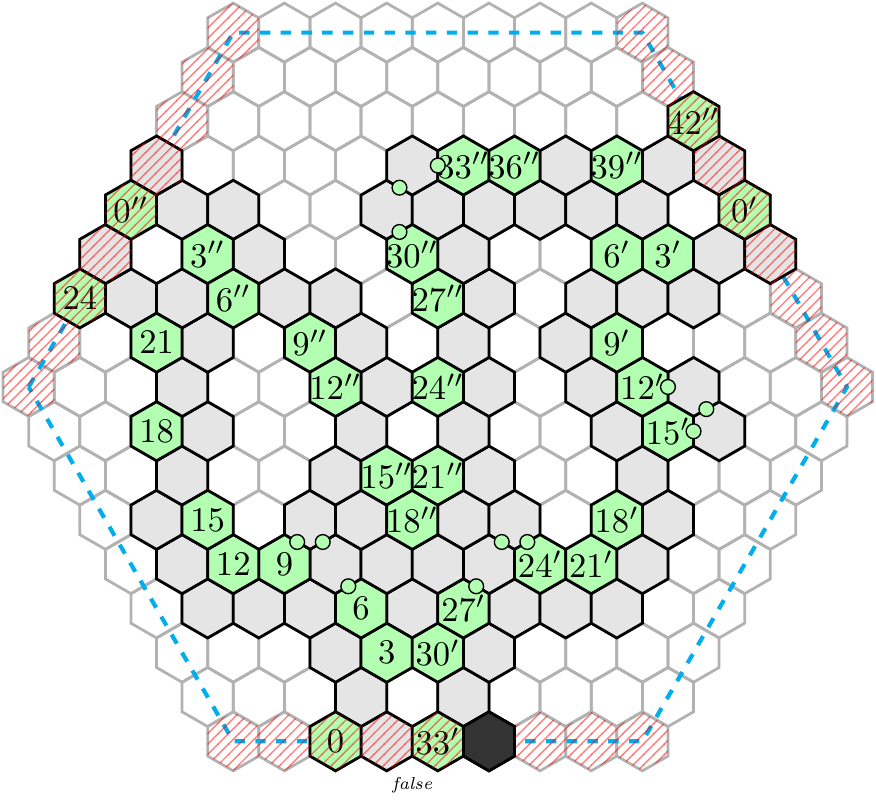}
  }
  \caption{
    Given an input bit on the bottom side, duplicate-macrocells have only solutions for placing the numbers from
    $0$ to $24$, from $0'$ to $33'$ and from $0''$ to $42''$, which copy the input bit as output bits to the two other sides.
    The {\em true} and {\em false} input bits consists in a cell on the bottom side
    which already hosts a number from the adjacent macrocell.
  }
  \label{fig:macrocell-duplicate}
\end{figure}

\begin{figure}[h!]
  \centerline{
    \includegraphics[scale=.75]{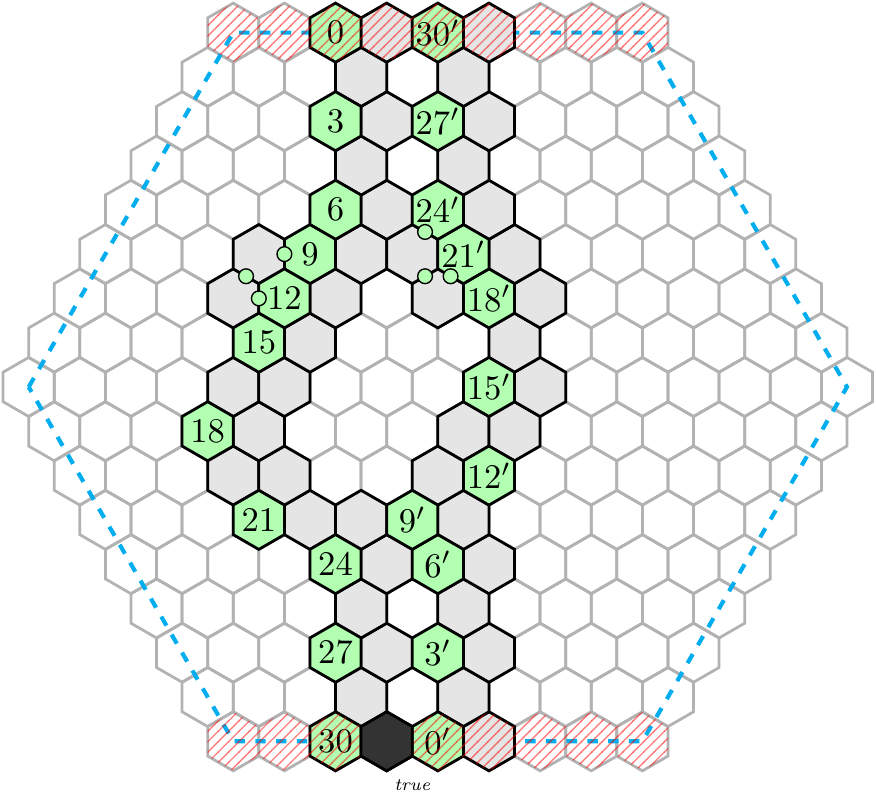}
    \includegraphics[scale=.75]{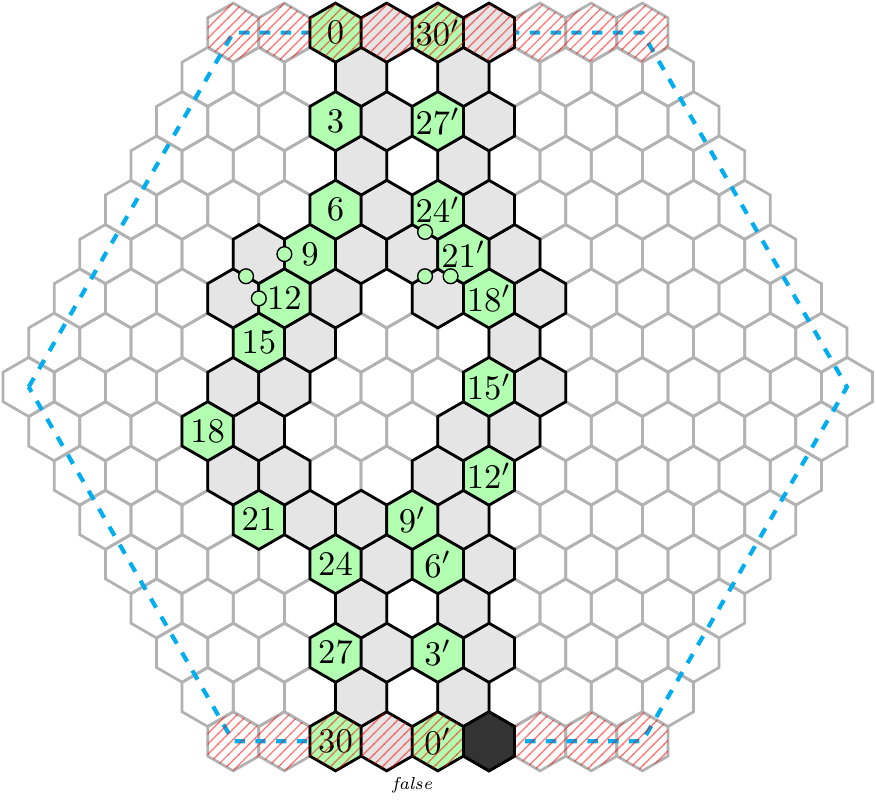}
  }
  \caption{
    Given an input bit on the bottom side, negation-macrocells have only solutions for placing the numbers from
    $0$ to $33$ and from $0'$ to $27'$, which flip this bit as an output bit to the top side.
    The {\em true} and {\em false} input bits consists in a cell on the bottom side
    which already hosts a number from the adjacent macrocell.
  }
  \label{fig:macrocell-negation}
\end{figure}

\begin{figure}[h!]
  \centerline{
    \includegraphics[scale=.75]{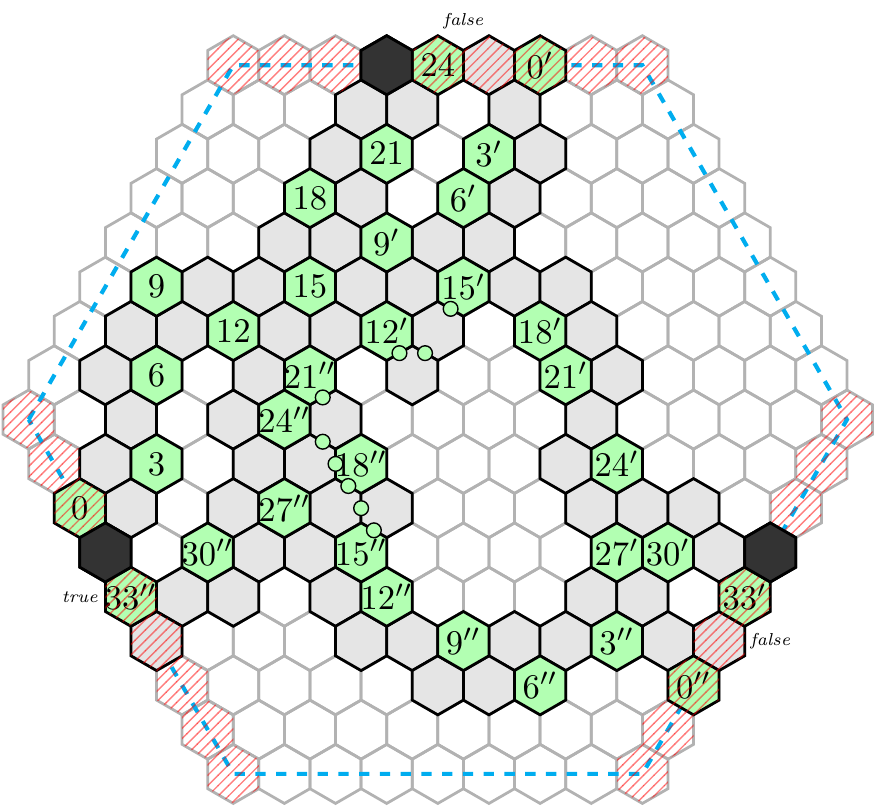}
    \includegraphics[scale=.75]{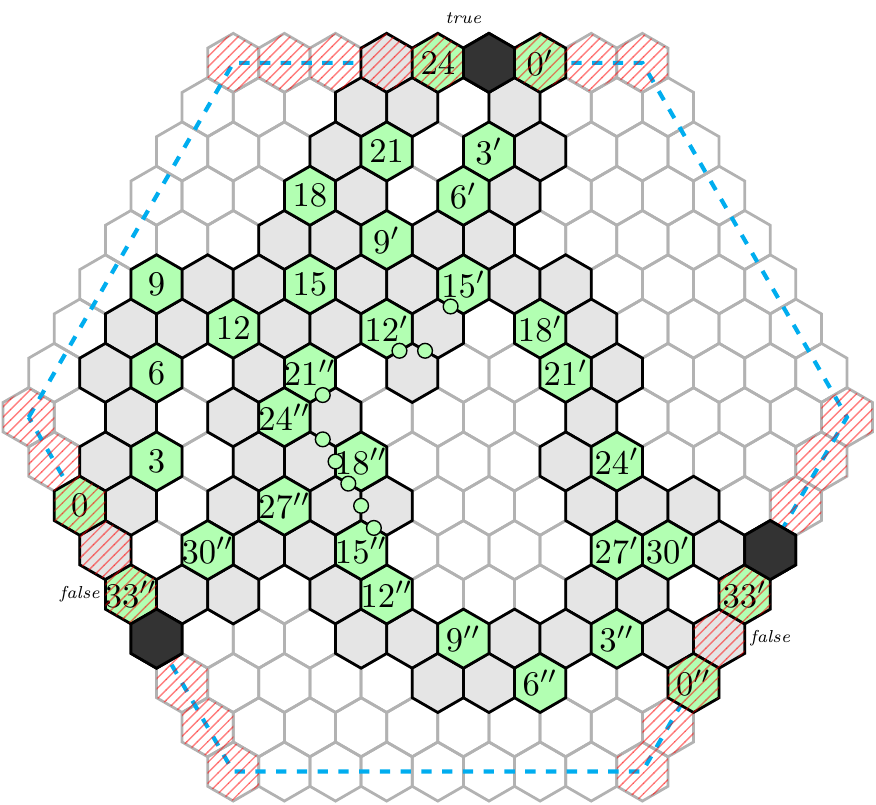}
  }
  \centerline{
    \includegraphics[scale=.75]{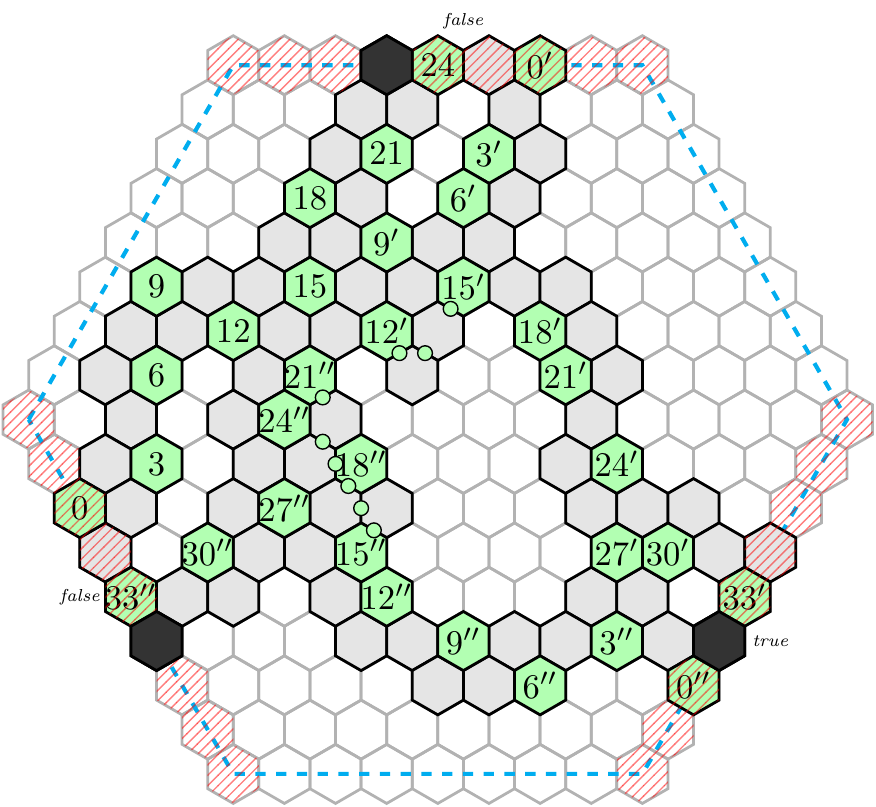}
    \includegraphics[scale=.75]{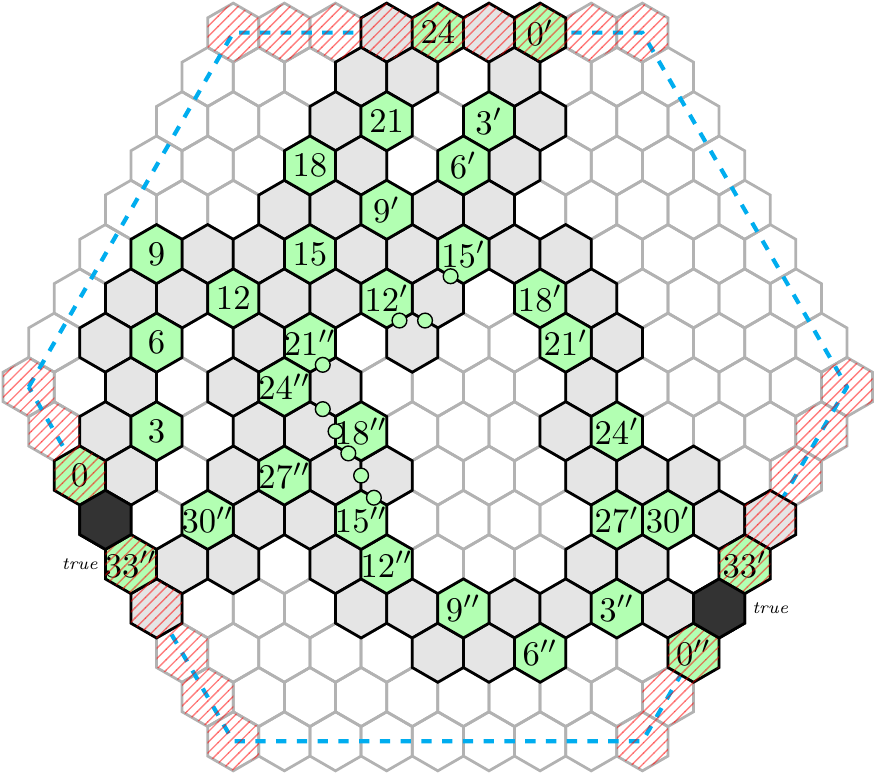}
  }
  \centerline{
    \includegraphics[scale=.75]{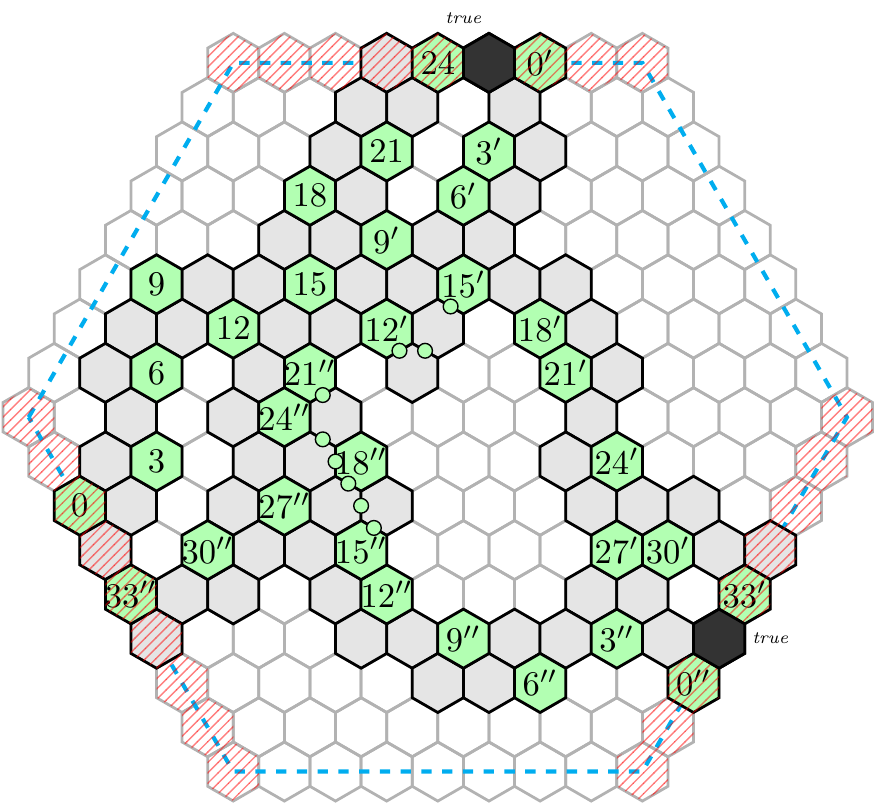}
    \includegraphics[scale=.75]{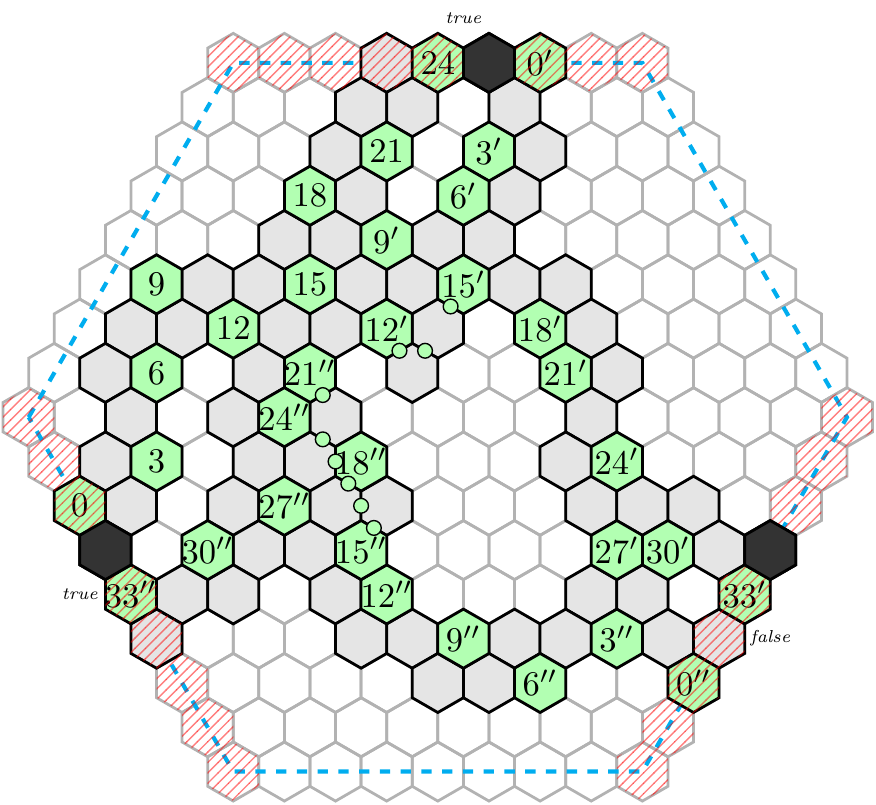}
  }
  \caption{
    Given three input bits on its sides, the clause-macrocell has a solution
    if and only if exactly one bit is {\em true}.
    The six pictures cover all the possibilities:
    when a bit is not indicated it means that regardless of its value the macrocell has no solution
    (because it already has the two other input bits set to {\em true}).
  }
  \label{fig:macrocell-clause}
\end{figure}

\begin{figure}[h!]
  \centerline{
    \begin{tikzpicture}
      \node at (0,0) {\includegraphics[scale=.5]{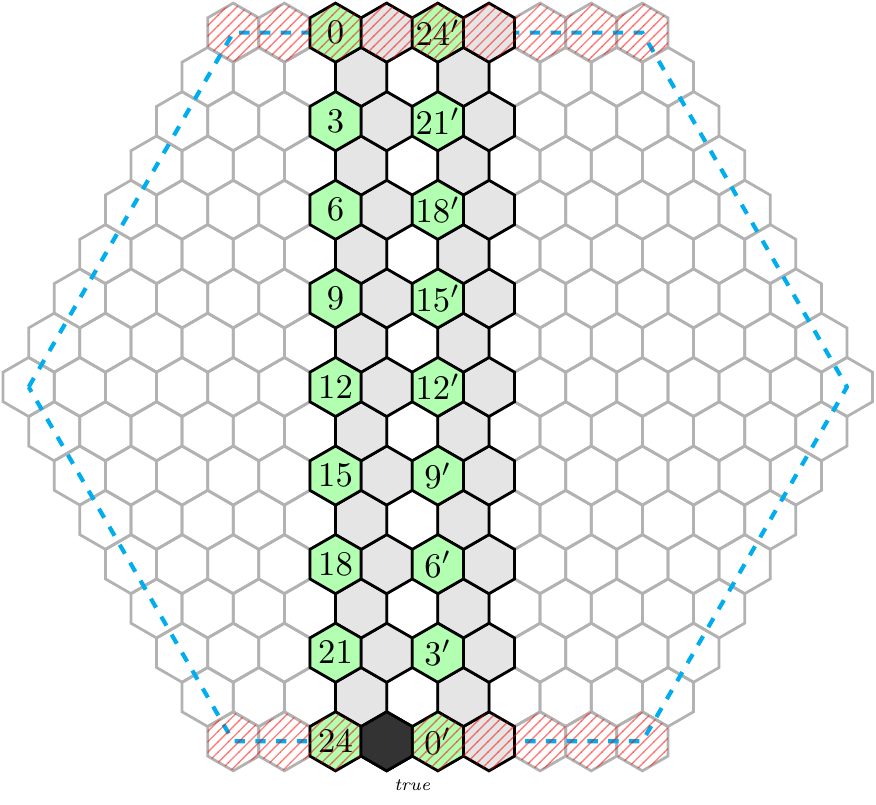}};
      \node at (6.6,0) {\includegraphics[scale=.5]{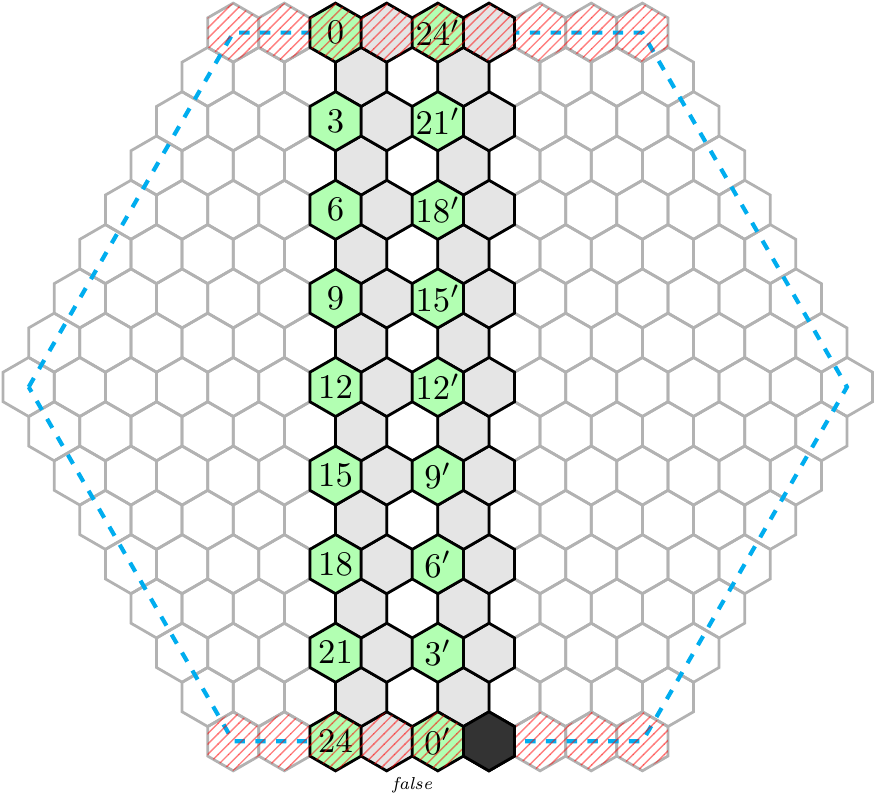}};
      \node at (3.3,-2.2) {\includegraphics[scale=.5]{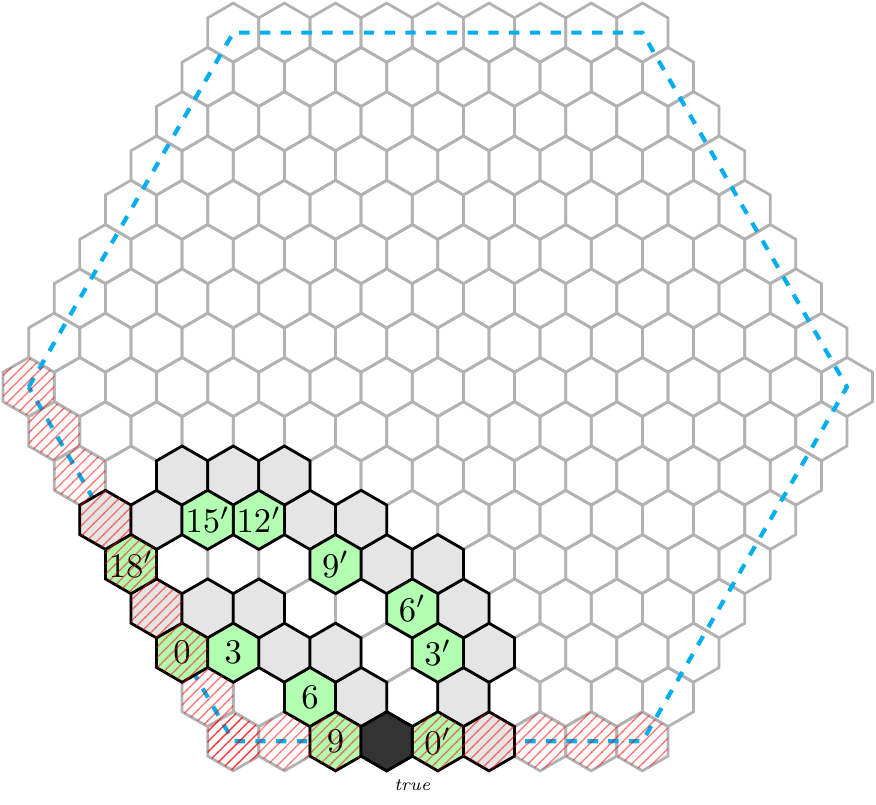}};
      \node at (9.9,-2.2) {\includegraphics[scale=.5]{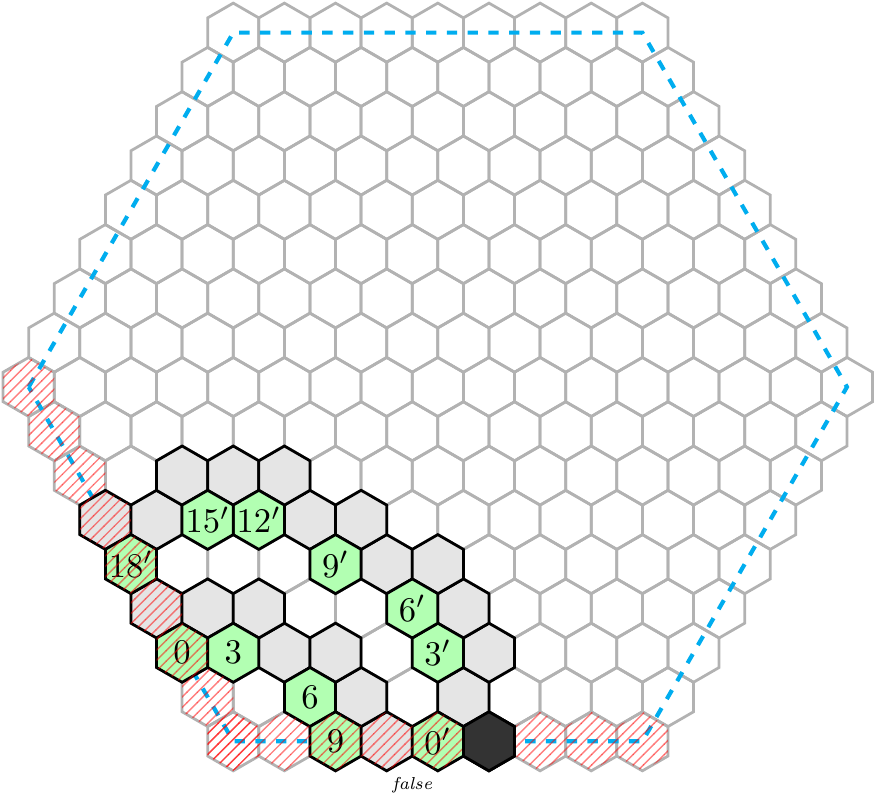}};
      \node at (0,-4.4) {\includegraphics[scale=.5]{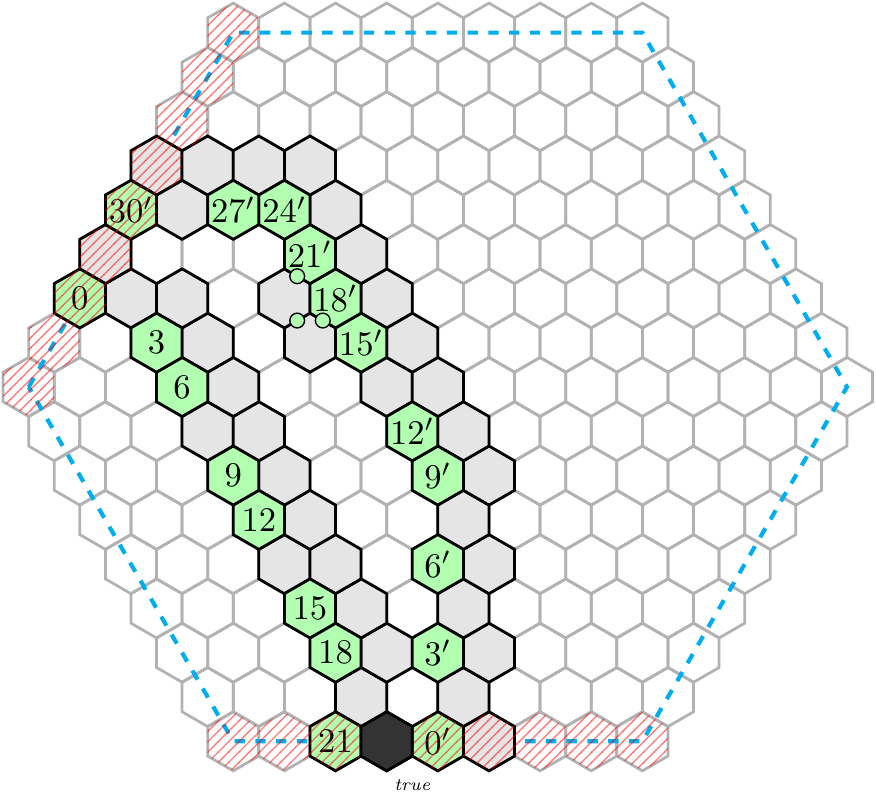}};
      \node at (6.6,-4.4) {\includegraphics[scale=.5]{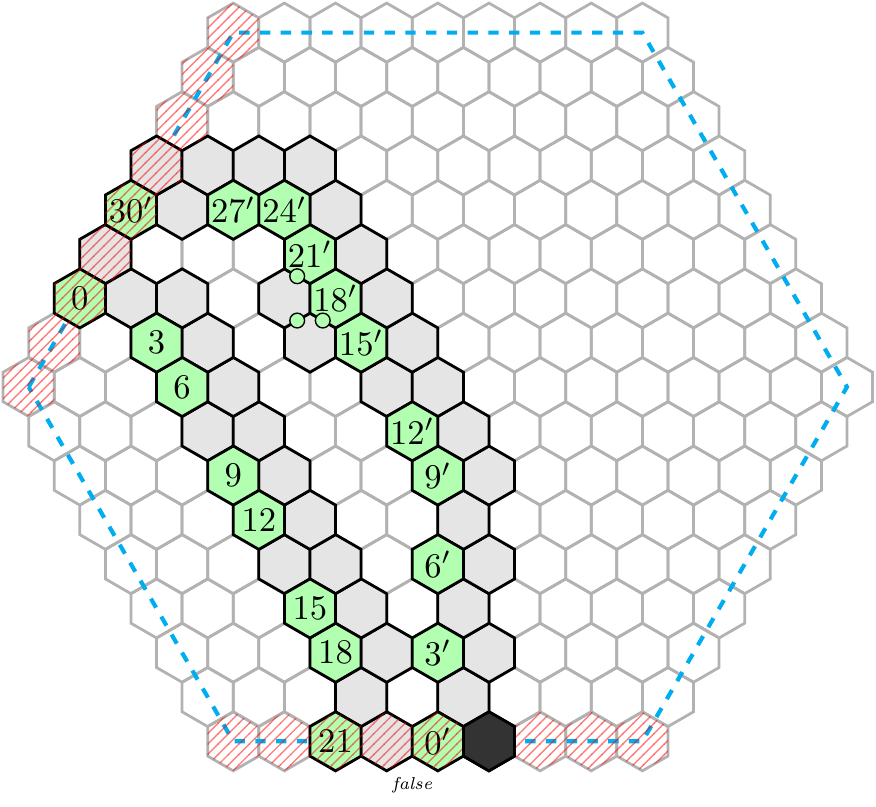}};
      \node at (3.3,-6.6) {\includegraphics[scale=.5]{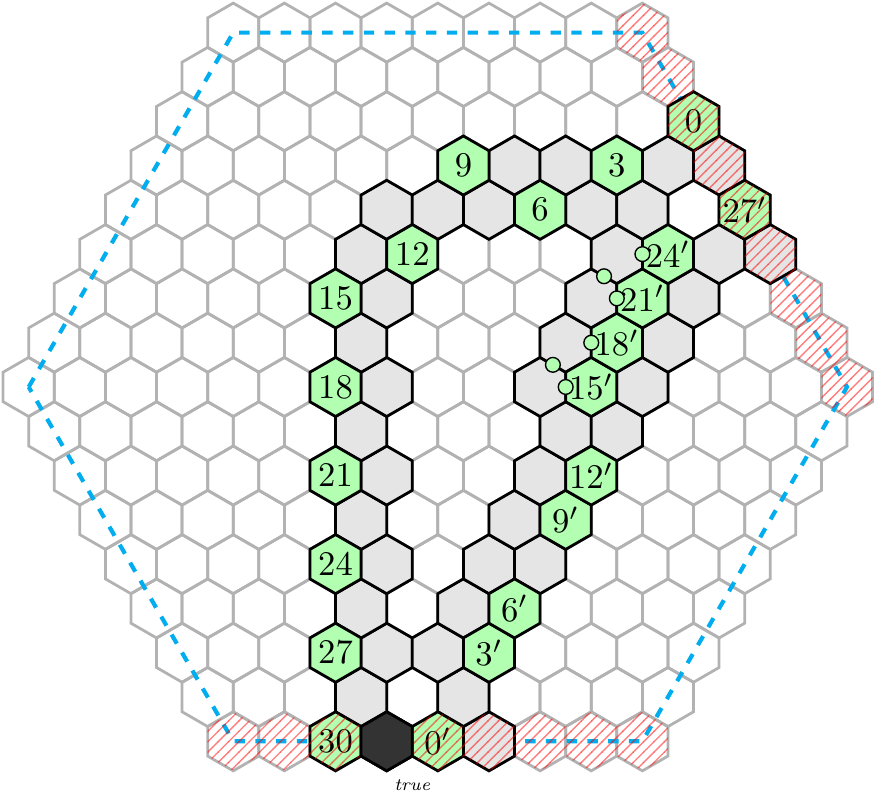}};
      \node at (9.9,-6.6) {\includegraphics[scale=.5]{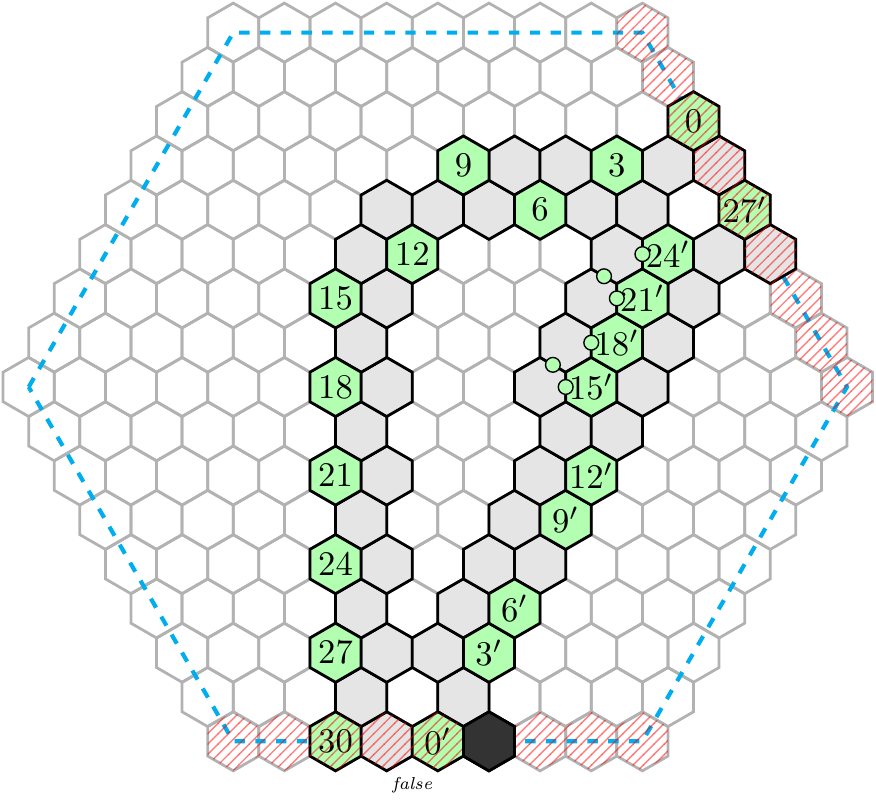}};
      \node at (0,-8.8) {\includegraphics[scale=.5]{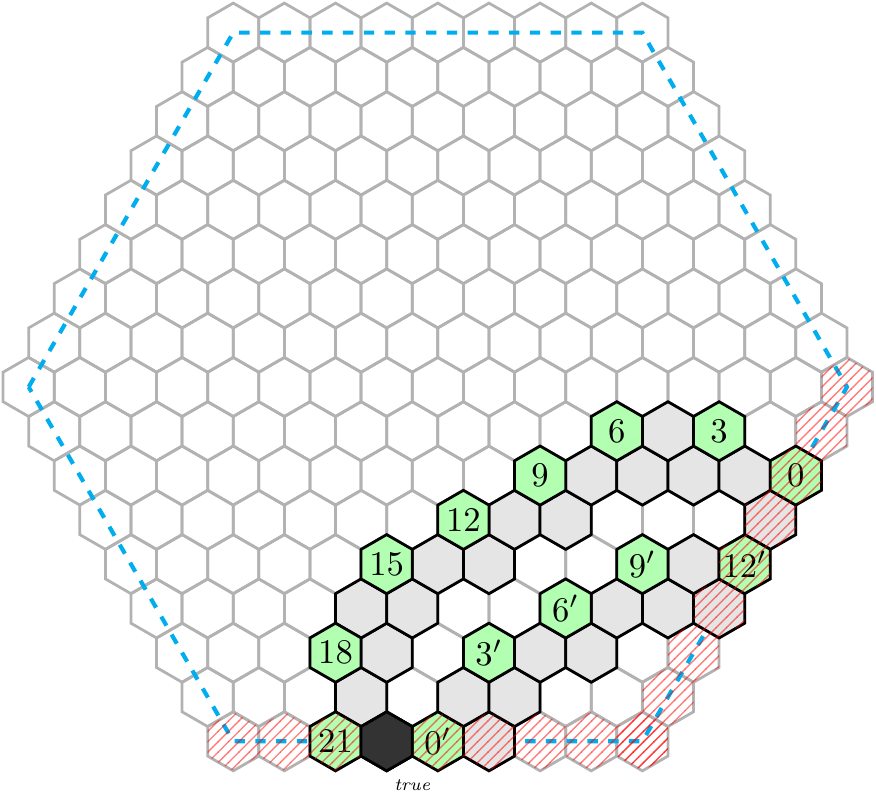}};
      \node at (6.6,-8.8) {\includegraphics[scale=.5]{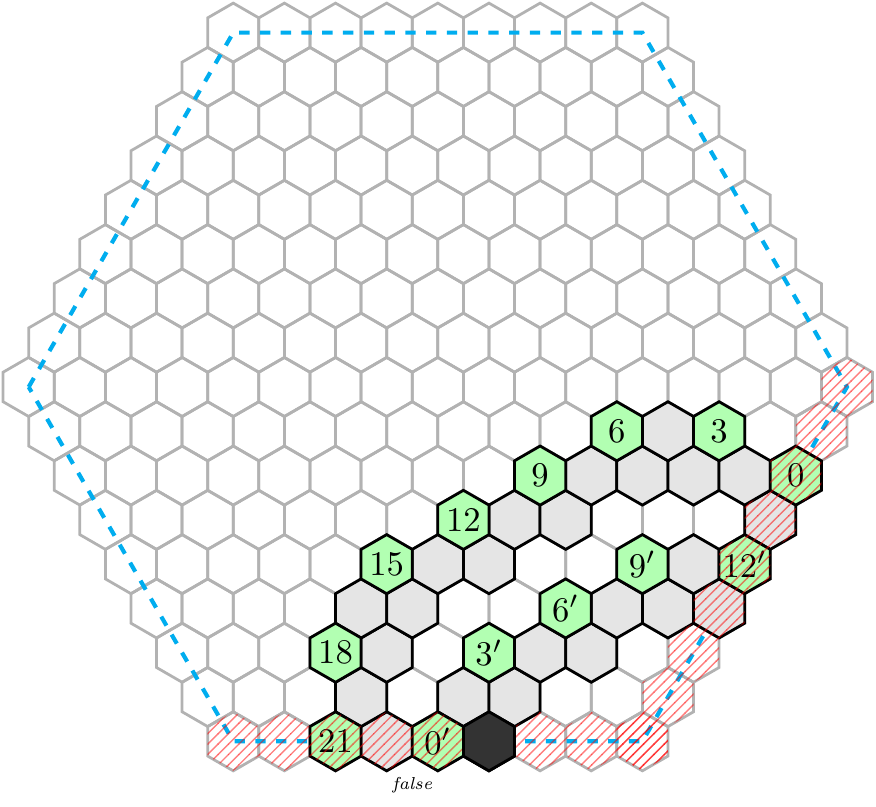}};
    \end{tikzpicture}
  }
  \caption{
    Given an input bit on the bottom side, each wire-macrocell has only solutions 
    which copy this bit as an output bit to the other side.
    The {\em true} and {\em false} input bits consists in a cell on the bottom side
    which already hosts a number from the adjacent macrocell.
  }
  \label{fig:macrocell-wires}
\end{figure}

\end{document}